\documentclass[journal,twoside,web]{ieeecolor}
\usepackage{generic}
\usepackage{cite}
\usepackage{amsmath,amssymb,amsfonts}
\usepackage{algorithmic}
\usepackage{subfig}
\usepackage{graphicx}
\usepackage{bm}
\usepackage{adjustbox}
\usepackage{pifont}
\usepackage{tikz}
\usepackage{transparent}
\usepackage{algorithm,algorithmic}
\usepackage{hyperref}
\usepackage{enumerate}
\usepackage{soul}
\usepackage{color}
\usepackage{mathtools}
\newtheorem{definition}{Definition}
\newtheorem{theorem}{Theorem}
\newtheorem{lemma}{Lemma}
\newtheorem{example}{Example}
\newtheorem{problem}{Problem}
\usepackage{tcolorbox}
\newtheorem{remark}{Remark}
\newtheorem{assumption}{Assumption}
\newtheorem{corollary}{Corollary}
\usepackage{csquotes}
\usepackage{mathtools}
\usepackage{titlesec}

\usepackage{pifont}
\newcommand{\cmark}{\ding{51}}%
\newcommand{\xmark}{\ding{55}}%
\hypersetup{hidelinks=true}
\usepackage{textcomp}
\usepackage{version}
\includeversion{arxiv}
\excludeversion{tac}

\def\BibTeX{{\rm B\kern-.05em{\sc i\kern-.025em b}\kern-.08em
    T\kern-.1667em\lower.7ex\hbox{E}\kern-.125emX}}
\usepackage{comment}

\begin{document}
\title{Network-aware Recommender System\\ via Online Feedback Optimization}
\author{Sanjay Chandrasekaran, Giulia De Pasquale, Giuseppe Belgioioso, and Florian D\"orfler
\thanks{S. Chandrasekaran, G. De Pasquale, F. D\"orfler are with the Automatic Control Laboratory, Department of Electrical Engineering and Information Technology, ETH Z\"urich, Physikstrasse 3 8092 Z\"urich, Switzerland  (e-mails: \texttt{\{schandraseka,degiulia,doerfler\}@ethz.ch}). G. Belgioioso is with the Division of Decision and Control Systems (DCS) at KTH Royal Institute of Technology, Malvinas v\"ag 10, SE-100 44 Stockholm, Sweden (e-mail: \texttt{giubel@kth.se}). This work was partially supported by the SNSF via NCCR Automation (Grant Number 180545) and by the Wallenberg AI, Autonomous Systems and Software Program (WASP) funded by the Knut and Alice Wallenberg Foundation.}}

\maketitle
\begin{abstract}
Personalized content on social platforms can exacerbate negative phenomena such as polarization, partly due to the feedback interactions between recommendations and the users.
In this paper, we present a control-theoretic recommender system that explicitly accounts for this feedback loop to mitigate polarization.
Our approach extends \textit{online feedback optimization} -- a control paradigm for steady-state optimization of dynamical systems -- to develop a recommender system that trades off users engagement and polarization reduction, while relying solely on online click data.
We establish theoretical guarantees for optimality and stability of the proposed design and validate its effectiveness via numerical experiments with a user population governed by Friedkin–Johnsen dynamics. Our results show these ``network-aware" recommendations can significantly reduce polarization while maintaining high levels of user engagement.

\end{abstract}


\section{Introduction}

Online social platforms use recommender systems to provide users with tailored content to maximize engagement over the platform. To do so, recommender systems exploit techniques, such as content-based filtering \cite{bansal} and collaborative filtering \cite{erinaki}, that combine information personalization, popularity and similarity of interests with other users.
A significant drawback of content personalization is related to the \emph{amplification} phenomenon \cite{KD-BS-KS-WU:21}, associated with the feedback loop that arises between recommended content and user preferences. Amplification leads to misinformation spread \cite{RUFFO2023100531}, filter bubbles \cite{Geschke}, bias exacerbation \cite{pagan2023classification} and, ultimately, polarization \cite{David,Bakshy}. 
Motivated by these challenges, we consider instead recommender systems inspired by control-theoretic principles that have the potential of explicitly accounting for the feedback interplay between recommendations and user interests to mitigate polarization. 
Our goals are to develop a systems-theory-based recommender design to understand i)~the impact of recommendations on user opinions, and ii)~how recommender systems should depart from engagement maximization to mitigate polarization. We approach the problem by building upon the systems-theory perspective of recommender systems pioneered in \cite{rossi,JC-SK-SKE-UW-SD-SI:24} and by leveraging on \textit{online feedback optimization} (OFO) \cite{hauswirth2021optimization} -- an emerging control paradigm for steady-state optimization of dynamical systems. Specifically, we aim at designing a recommender system as a dynamic feedback controller that mitigates polarization by providing users personalized content, using solely \emph{implicit feedback} from their online click data, as illustrated in Figure \ref{fig:blockdiagram_1}.

\begin{figure}
\centering
\includegraphics[width=\columnwidth]{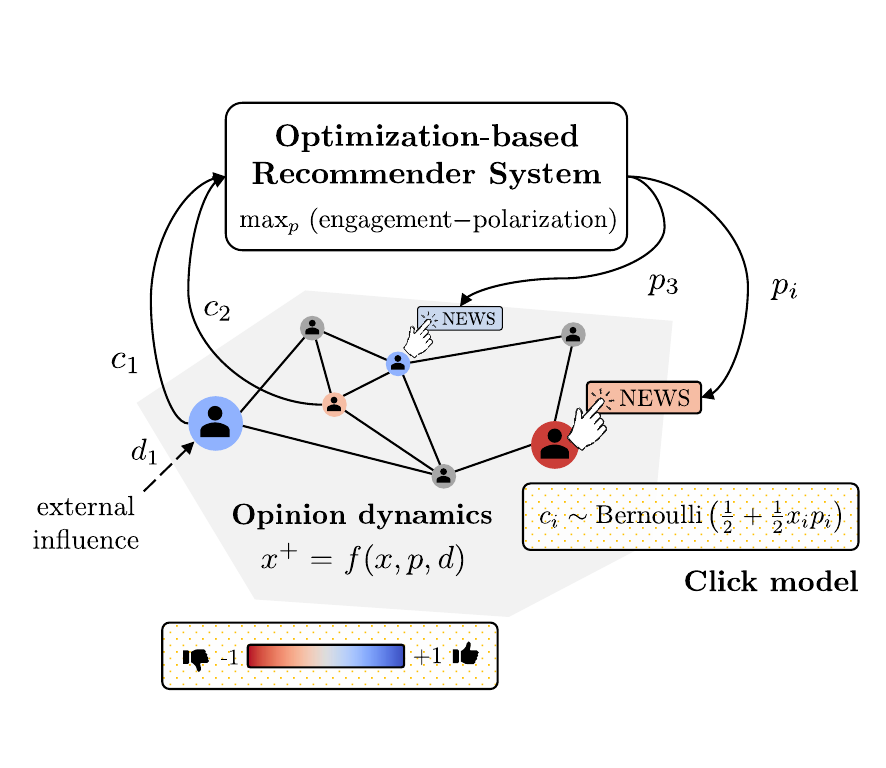}
    \caption{Closed-loop interconnection between users' opinion in a social network and the proposed recommender system.}
    \label{fig:blockdiagram_1}
\end{figure}

\subsection*{Related Work}
\textbf{Recommender systems.} Most of the literature on recommender systems  relies on static methods (such as matrix factorization) and neglects the dynamic nature of the feedback interactions between the users and the online platforms \cite{Chen}. In contrast, our work develops along a more recent systems-theory perspective on recommender systems \cite{rossi, JC-SK-SKE-UW-SD-SI:24}, where the feedback interaction between social platform and recommender systems is made explicit.
In \cite{rossi}, the recommender system provides only extreme positions to the users based on their clicking history.  The authors show that when \emph{confirmation bias} \cite{nick} -- the tendency to engage with content that aligns with prior beliefs -- is the sole factor driving clicks, their recommender system promotes opinion polarization. The recommender design proposed in \cite{Lanzetti} alternates between exploitation and exploration by either providing news which more closely align with the users' beliefs or exploring new content. The authors show that shifts in the opinions at the individual level do not necessarily reflect in the opinion distribution of the population.  In \cite{sprenger2024control}, the recommender system is modeled as an artificial user within the social platform, providing news content at minimum distance from all users opinions to promote engagement. This study investigates how such a recommender system can affect opinion dynamics across the network. 
 In \cite{JC-SK-SKE-UW-SD-SI:24}, the focus is on the interplay between a single user, the recommender system, and harmful content. The authors seek recommendation policies that balance maximizing the click-through rate (CTR) and mitigate harm. Their findings suggest that to properly mitigate harm, it is necessary to account for user dynamics. 
Notably, most of these studies focus solely on individual-user designs, where recommendations are tailored to each user in isolation, without explicitly accounting for the broader influence of social interactions. However, social interactions are well-known to be a crucial factor in shaping opinions over a connected population \cite{FJ,DG,polaropiniondynamics}.

\textbf{Online feedback optimization.} Most OFO controllers are based on first-order optimization algorithms \cite{hauswirth2021optimization}, and their deployment require real-time access to the states of the plant, an explicit knowledge of its input-output sensitivity (or a closed approximation), and a closed-form expression of the gradient of the control objective.
However, some recent extensions \cite{picallo,10278489, Cothren,cothren2,zhiyu} have relaxed some of these stringent requirements. Specifically, \cite{picallo} proposes to estimate the input-output sensitivity of the plant from real-time state measurements using recursive-least square; while \cite{10278489} achieves similar results leveraging historical input-output trajectories and behavioral systems theory. In \cite{Cothren}, the authors address the case where exact states measurements cannot be directly accessed, but need to be inferred from perceptual information. Their solution relies on the integration of a neural network component that transform high-dimensional sensory data into state estimates. Similarly, \cite{cothren2} proposes a neural-network solution for the case where a closed-form expression of the gradient is not available. Finally, \cite{zhiyu} utilizes a zeroth-order method to estimate the composite gradient directly from online cost function evaluations, eliminating the need for both input-output sensitivity and a closed-form gradient expression.
See Table \ref{literaturecomparisontable} for a comparative overview of these OFO controllers. 


\begin{table}[t]
    \begin{center}
    \caption{Comparative overview of the main requirements of some recent online feedback optimization (OFO) controllers.}\label{literaturecomparisontable}
    \begin{tabular}{c c c c c c }
    
        Controller & Sensitivity & States & Objective & Gradient & Convex\\
        \hline \hline
        \cite{hauswirth2021optimization} & \cmark & \cmark & \cmark & \cmark & \cmark\\ 
        \hline
        \cite{picallo, 10278489} & \xmark & \cmark & \cmark & \cmark & \cmark\\
        \hline
        \cite{Cothren} & \cmark & \xmark & \cmark & \cmark & \cmark\\
        \hline
        \cite[Alg. 2]{cothren2} & \cmark & \cmark & \xmark & \xmark & \cmark\\
        \hline
        \cite{zhiyu} & \xmark & \cmark & \cmark & \xmark & \xmark\\
        \hline
        This work & \xmark & \xmark & \xmark & \xmark & \xmark\\
        \hline
    \end{tabular}
    \end{center}
\end{table}

\subsection*{Contributions}

Unfortunately, existing OFO designs cannot be directly applied to social networks, as recommender systems typically only have access to implicit feedback in the form of click data. Real-time users opinion are inaccessible, social interactions may be hidden due to privacy preferences, and there is no closed-form expression for the expected CTR as it depends on unknown users clicking behavior -- namely, their tendencies or pattern when interacting with recommended items.

In this paper, we address these challenges by integrating auxiliary estimation blocks into a projected-gradient OFO controller to compensate for the missing information, thereby enabling online feedback optimization in social networks.  Specifically, we estimate opinions and clicking behaviour using supervised learning, sensitivity using Kalman filtering and gradients using zero-order methods. 

Unlike most system-theory-based recommender system designs in the literature, our proposed algorithm is \textit{network-aware}, namely, it leverages on users' social interactions to better mitigate undesired population phenomena such as polarization, while maintaining high levels of engagement.

Our contribution is threefold:
\begin{enumerate}[(i)]
    \item We propose a novel OFO-based recommender system algorithm that steers a social network towards steady-state opinion profiles that balance between maximizing overall click-through rate and minimizing polarization. Our design is agnostic about real-time opinions of the users, their network interactions and clicking behavior, and solely relies on click data measured online.

    \item We establish closed-loop stability guarantees of the recommender system-social network interconnection with respect to locally-optimal solutions of the steady-state control objective, assuming exponential stability of the opinion dynamics and other mild technical conditions.
    
    \item We validate the proposed recommender system through numerical simulations on a social network where user opinions are governed by an extended version of the Friedkin–Johnsen (FJ) dynamics \cite{rossi}, which incorporates the influence of recommended content. Our results demonstrate that network-awareness plays a crucial role in reducing opinion polarization in social platforms.

\end{enumerate}

\subsection*{Outline}
In Section \ref{sec:problemformulation}, we formulate the recommender objectives as a steady-state control problem. In Section \ref{sec:recommenderdesign}, we describe the recommender design using OFO. Section \ref{sec:results} present numerical simulations to evaluate the performance of our algorithm and highlight the benefits of network-awareness. Finally, Section \ref{sec:conclusions} summarizes our findings and concludes the paper.
\subsection*{Notations and Preliminaries}
We let the symbols  $\mathbb{R}$($\mathbb{R}_+$), $\mathbb{N}_0$ denote the set of (positive) real numbers and non-negative integers, respectively. The set of integers $\{1,2,\ldots,n\}$ is denoted by $[n]$. For a vector $y\in\mathbb{R}^n$, we let the symbol $|y|$ denote the vector whose $i$-th entry, $|y|_i$, is the $i$-th component of $y$ in modulus, $|y_i|$. The symbols $\lVert \cdot \rVert$, $\lVert \cdot \rVert_\infty$ denote the Euclidean and infinity norm, respectively. The symbol  $\langle \cdot, \cdot \rangle : \mathbb{R}^{n}\times  \mathbb{R}^{n}\rightarrow\mathbb{R}$ denotes the standard inner product on $\mathbb{R}^{n}$. The value of a signal $x$ at time instant $k$ is denoted by $x^k$. The symbol $1_n,(0_n)$, denotes the all ones (zeros) vector of size $n$. The symbols $I_n$ and $O_n$, denote the $n$-dimensional identity and zero matrix, respectively. Given a matrix $A\in\mathbb{R}^{m\times n}$, the matrix is \emph{positive} (\emph{non-negative}) if all its elements are greater (or equal) than zero, and we denote it as $A> 0$ ($A\geq 0$). We let $\mathrm{vec}(A)$ denote the vectorized version of $A$, i.e., if $A = [a_1 \ a_2 \  \ldots \ a_n]$, with $a_j\in\mathbb{R}^m, j\in[n]$, then $\mathrm{vec}(A) = [a_1^\top a_2^\top \ldots a_n^\top]^\top$. For a matrix $A\in\mathbb{R}^{n\times n}$, we denote by $\mathrm{det}[A]$ its determinant; $A$ is row-stochastic if $A\geq 0$ and $A1_n=1_n$.  The symbol $\mathcal{D}_\mathcal{U}\subseteq \mathbb{R}$ denotes the set of $n$-dimensional diagonal matrices with diagonal entries in $\mathcal U$. The symbol $\mathrm{diag}[x]$ denotes the diagonal matrix with $x_i$ on its $i$-th diagonal entry. Given $\mathcal{C} \subset \mathbb{R}^n$ and $x\in \mathbb{R}^n$, we let $\Pi_\mathcal{C}[x]$ denote the Euclidean projection of $x$ over $\mathcal{C}$. The normal distribution is denoted by $\mathcal{N}(\mu,\Sigma)$, with $\mu, \Sigma$ representing the mean and variance, and the uniform distribution over the interval $[a,b]$, with  $a,b, \in \mathbb{R}$ is represented by $\mathcal{U}[a,b]$.

We now state some preliminary lemmas and definitions that are recurrently used in this manuscript. 
\begin{definition}[Minimal modulus of continuity\cite{BRENEIS}]\label{modcontinuity}
    Let $h:\mathcal{X}\to\mathbb{R}$ be a continuous function over a  bounded set $\mathcal{X}$. The \emph{minimal modulus of continuity} of $h$ on $\mathcal{X}$ is defined as%
    $$\omega_h(\gamma) := \mathrm{sup}\Big\{|h(x)-h(y)|:x,y\in\mathcal{X}, \lVert x-y \rVert_\infty \leq \gamma\Big\}.$$
\end{definition}

\smallskip
\begin{definition}[Persistently exciting input \cite{WILLEMS}] \label{POEhenkel}
The signal $u\in\mathbb{R}^n$ 
is said to be persistently exciting of order $n$ if there exists a time span $S>0$ such that for all $k>0$, the matrix formed by the columns $u^{k+i}\in\mathbb{R}^n$ for $i\in\{0,1,\ldots S\}$ has full rank, i.e. $\mathrm{Rank}[u^k \ u^{k+1} \ \ldots \ u^{k+S}] = n$.
\end{definition}
\begin{definition}[Global $\beta$-smoothness  \cite{nesterov}]
     A continuously differentiable map $\Phi: \mathbb{R}^n \rightarrow \mathbb{R}$ is globally $\beta$-smooth if its gradient $\nabla \Phi$ is globally $\beta$-Lipschitz.
\end{definition}
\begin{lemma}[Properties of $\beta$-smooth functions \cite{nesterov}]\label{Lsmoothlemma}
    Given a globally $\beta$-smooth map $\Phi: \mathbb{R}^n \rightarrow \mathbb{R}$, then $\Phi(x_1) - \Phi(x_2) - \nabla \Phi^\top (x_2) (x_1-x_2) \leq \frac{1}{2}\beta\lVert x_1-x_2\rVert^2$. Further, if $\Phi$ is twice continuously-differentiable, then $\lVert\nabla^2 \Phi\rVert \leq \beta$.
\end{lemma}

\section{Problem setup}
\label{sec:problemformulation}
\subsection{Social Network Modelling}
We consider a recommender system for a social network consisting of $n$ users, indexed by $i \in [n]$ that makes sequential, personalized recommendations. User opinions are collected into a vector $x\in [-1,1]^n$, with $x_i$, $i\in [n]$, denoting the opinion of the $i$-th user. The users' opinions evolve as
\begin{equation}\label{opmodel}
x^{k+1}  = f(x^k,p^k,d),
\end{equation}
where  $p\in [-1,\ 1]^n$ refers to the positions that the recommended article takes on, whose $i$-th entry is the position provided to the $i$-th user, $d\in [-1,\ 1]^n$ represents an external influence to the platform, e.g., television, newspapers or a personal bias, and $f:([-1,\ 1]^n)^3\to [-1,\ 1]^n$ encodes the influence of interactions among users. We assume there is a single topic of discussion or issue, on the platform. For the sake of illustration, if the issue is polarizing, then, a recommendation with position $p=+1$ ($-1$) can be interpreted as news strongly in support of (against) the issue.

The following assumption ensures that the dynamics \eqref{opmodel} are well-posed and admit a steady-state mapping.
\begin{assumption}[Well-posedness]\label{ass:ODasm}
The following hold:
\begin{enumerate}
\item[(i)] The dynamics \eqref{opmodel} are forward invariant in $[-1,\ 1]^n$, i.e.,
\begin{equation*} \label{eq:FIDyn}
x^0, p^0, d^0\in [-1,\ 1]^n \, \Rightarrow \, x^k \in [-1,\ 1]^n,  \;  \forall k\in\mathbb{N}_0.
\end{equation*}

\item[(ii)] The dynamics \eqref{opmodel} are uniformly exponentially stable with constant $p,d$, and admit a unique steady-state map $h:  \big([-1,\ 1]^n\big)^2 \rightarrow [-1,\ 1]^n$ satisfying $$h(p,d) = f(h(p,d),p,d),\quad \forall p,d \in [-1,\ 1]^n.$$

\item[(iii)] The map $h(p,d)$ is continuously-differentiable and $L-$Lipschitz \cite{nesterov} with respect to $p$.
{\hfill $\square$}
\end{enumerate}
\end{assumption}

This assumption can be easily adopted for a wide array of renowned opinion dynamics models, such as De-Groot model \cite{DG}, FJ model \cite{FJ}, and polar opinion dynamics model \cite{polaropiniondynamics}, among others, for which the influence of a recommender system can be, in principle, included. For an extensive list of such models, we refer the reader to the survey \cite{PROSKURNIKOV}. Traditionally, in fact, these models  consider social interactions and prejudice as the only driving forces behind opinion formation. Nonetheless, recommendations on contentious topics provably influence user opinions \cite{rossi}, underscoring the need to include their influence in the opinion dynamics. Next, we illustrate an example of how the FJ model, one of the most widely used opinion dynamics models across different communities \cite{LZ-QB-ZZ:21,AM-ET-PT:17,CM-CM-CET:18,BD-KJ-OS:15,AG-ET-PT:13,rossi,JC-SK-SKE-UW-SD-SI:24}, can be extended to capture the effects of the recommender on networked users.  
\begin{example}[Extended Friedkin--Johnsen model]\label{ex:FJ}
The opinions evolve according to
\begin{equation} \label{eq:FJdyn}
    x^{k+1} = (I_n - \Gamma_p - \Gamma_d)Ax^k + \Gamma_p p^k + \Gamma_d d,
\end{equation}
where $\Gamma_d, \Gamma_p$ are positive diagonal matrices such that $\Gamma_p + \Gamma_d \preceq I_n$, describing the impact of $d$ and $p$ over the opinions, and $A$ is a row-stochastic adjacency matrix encoding the social interconnections of the users. 
The dynamics \eqref{eq:FJdyn} are forward-invariant in $[-1,1]^n$, since the opinions are a convex combination of $x,p,d\in[-1,1]^n$.
Further, for a constant position provided by the recommender system, the steady-state mapping is single-valued, affine (hence, continuously differentiable), and reads as
\begin{equation}
h(p,d) = (I_n - (I_n - \Gamma_p - \Gamma_d)A)^{-1} (\Gamma_p p + \Gamma_d d ).\notag
\end{equation}
%
Note that by taking $\Gamma_p = O_n$ and $d=x^0$ to be an initial bias, the dynamics \eqref{eq:FJdyn} reduce to the standard FJ model \cite{FJ}.
{\hfill $\triangle$}
\end{example}

Engagement of the users over the platform is measured in terms of CTR, measured as number of clicks over number of provided news articles.
As common practice in the context of \emph{implicit feedback} recommender systems \cite{MW-MG-XZ-KZ:18}, for a generic user $i \in [n]$ holding an opinion $x_i$, we model the probability of $i$ clicking on a recommendation with position $ p_i$ as a random variable $c_i$ sampled from a Bernoulli distribution with unknown argument $g_i(p_i,x_i)$, i.e.  $c_i \sim \mathcal{B}(g_i(p_i,x_i))$. The argument $g_i(p_i,x_i)$ represents the \textit{clicking behaviour} of user $i$ and depends on the user's opinion $x_i$ as well as the position of the recommendation $p_i$. We will work under the following regularity assumptions on the clicking behaviour:
\begin{assumption}[Clicking behaviour]\label{measurementmodelassumpt}
The clicking~behaviour $g(p,x): [-1,1]^n\times[-1,1]^n \to [0,1]^n $ is
$M_x$~-~Lipschitz with respect to $x$, and $L_p$- and $L_x$-smooth with respect to $p$ and $x$, respectively.
{
  \hfill $  \square$
    }
\end{assumption}

Next, we provide an example of a clicking behaviour.
\begin{example}[Extremity confirmation bias \cite{rossi}]\label{ex:clickingbehaviour}
A user $i~\in~[n]$, holding opinion $x_i$, affected by \emph{extremity confirmation bias} clicks on a recommendation $ p_i$ with probability 
\begin{equation}\label{eq:CA}
    c_i \sim \mathcal{B}\left(\frac{1}{2} + \frac{1}{2}x_i p_i \right),\notag
\end{equation}
where the clicking behaviour $g_i(p_i,x_i) = \frac{1}{2} + \frac{1}{2}x_ip_i$ models confirmation bias towards extreme recommendations. In fact, for an opinion $x_i \approx \pm 1$  and a position $p_i \approx \pm 1$ ($\mp1$), the probability of clicking is almost $1$ ($0$). The click event becomes random  as the user position approaches the neutral stance $p_i \approx 0$,  highlighting diminished engagement for less polarized content or when recommendations directly counter the user's stance on the issue.
{\hfill $\triangle$}
\end{example}

\subsection{Problem formulation}
The goal of the recommender system is to provide recommendations that optimize a specific metric, denoted as $\varphi(p,x)$. Typical choices in the literature include maximizing CTR \cite{rossi,JC-SK-SKE-UW-SD-SI:24}, opinion control \cite{opinioncontrol}, and  polarization mitigation \cite{polarizationcontrol}. Deviating from traditional recommender systems, we consider a multi-objective cost function combining engagement maximization and polarization mitigation, given by  
\begin{equation} \label{eq:obj_rec}
\varphi(p,x) = \varphi^{\text{ctr}}(p,x) + \gamma \, \varphi^{\text{pol}}(x),
\end{equation}
where $\gamma$ is a non-negative trade-off parameter.
The engagement-related term in \eqref{eq:obj_rec} is the CTR, formalized as
\begin{equation*}
\varphi^{\text{ctr}}(p,x) = -\sum_{i \in [n]} \mathbb{E}_{c_i \sim \mathbb{B}(g_i(x_i,p_i))}[\,c_i\,],
\end{equation*}
where $\mathbb{E}_{c_i \sim \mathbb{B}(g_i(x_i,p_i))}[\,c_i\,]$ is the expectation of user $i$'s clicking, given their opinion $x_i$ and a recommendation with position $p_i$.
The second term in \eqref{eq:obj_rec}, $\varphi^{\rm pol}$, is a smooth and Lipschitz function that measures opinion polarization. For our algorithm we consider it to be $\varphi^{\text{pol}}(x) := \sum_{i \in[n]} s_i(x)$, where $s_i$ penalizes deviations from the moderate (zero) opinions:
\begin{align}\label{polarizationcost}
s_i(x_i) = \begin{cases}
            (x_i - \epsilon_1)^2 & x_i<\epsilon_1\\
            0 & \epsilon_1 \leq x_i \leq \epsilon_2\\
            (\epsilon_2 - x_i)^2 & x_i>\epsilon_2
        \end{cases} 
\end{align}
Therein, the parameters $\epsilon_1 \leq \epsilon_2$ are used to control the degree of penalty towards extreme opinions. Specifically, a smaller positive (negative) choice for $\epsilon_2$ ($\epsilon_1$) indicates a higher penalty on extreme positive (negative) opinions. According to \cite{Bramson}, the metric in \eqref{polarizationcost} falls in the class of polarization spread, indicating how extreme the opinions are. 
Note that, for $\epsilon_1=\epsilon_2=0$, the term \eqref{polarizationcost} reduces to $\varphi^{\text{pol}}(x) = \lVert x \rVert^2$, another classic measure for opinion polarization \cite{AM-ET-PT:17}. 

The recommender system aims at finding the recommendation that drives the user population with dynamics \eqref{opmodel} to a steady-state minimizing the multi-objective cost \eqref{eq:obj_rec}:
\begin{subequations}
\label{opt}
\begin{align} 
    \underset{p,\, x}{\mathrm{minimize}} &\quad   \varphi(p,x) \label{opt_cost} \\
    \label{eq:ssconstr}
    \textrm{s.t.}  & \quad x = h(p,d) \\
                  & \quad p \in [-1,1]^n
\end{align}
\end{subequations}

This problem is uncertain and potentially non-convex, since the analytical form of the clicking behaviour is unknown and likely complex. Likewise, the opinion dynamics \eqref{opmodel} are unknown.
If the opinion dynamics in \eqref{opmodel} and the clicking behavior $g_i$ were known, and an accurate prediction of external influences $d$ were available, then one could aim to solve \eqref{opt} offline. However, in practice, none of these information is readily available, and the recommender system can provide personalized positions solely relying on clicks $\{c_i^k\}$.
\begin{remark}\label{rem:stationarity}
    Equation \eqref{eq:ssconstr} reflects the assumption that users' opinions reach steady-state before the recommender system takes action. This approach is in line with most of the existing works on recommender systems that incorporates preference dynamics \cite{JC-SK-SKE-UW-SD-SI:24,rossi,SD-JM:22,pagan2023classification,NEL-KL-AB-AF-YL:21}, where recommendation decisions are made based on the steady-state behavior of the users' preferences under the model.
\end{remark}


Our feedback control problem is formally stated as follows:
\begin{problem} \label{pr:problem1}
Design a feedback controller $p(c)$ so that \eqref{opmodel} converges to a solution $(p^*,x^*)$ of the optimization problem \eqref{opt}, assuming that only clicks $c_i^k$ are available. The clicking behaviours $g_i(p_i,x_i)$, steady-state map $h(p,d)$, the opinion dynamics \eqref{opmodel}, external influence $d$ are unknown, and the real-time opinions $x$ are not measurable.
\end{problem}

\section{Recommender system design}\label{sec:recommenderdesign}
We approach Problem \ref{pr:problem1} by designing a dynamic feedback controller inspired by the projected-gradient descent algorithm in \cite[Section 3A]{belgioioso2021sampled}.
The resulting recommender system dynamically generates recommendations as
\begin{equation}\label{PGD_ones}
  p^{k+1} = \Pi_{[-1,1]^n} \big[p^k - \eta\, \Phi(p^k,x^{k+1})\big], \quad \forall k \in \mathbb{N}
\end{equation}
where $\eta$ is a tunable controller gain, and 
\begin{align} \label{eq:grad}
\Phi(p,x) := \nabla_p \varphi(p,x) + \nabla_p h(p,d)^\top \nabla_x \varphi(p,x)
\end{align}
represents the gradient obtained by applying the chain-rule of differentiation to the cost $\varphi(p,x)$ in \eqref{opt}, with opinions at steady state, i.e., $x=h(p,d)$. In practice, evaluating the gradient \eqref{eq:grad} at each sampling instant $k$ requires access to:
\begin{enumerate}
\item[(i)] Real-time users' opinions $x^k$;
\item[(ii)] Sensitivity mapping $\nabla_p h(p,d)$; and
\item[(iii)] Gradients $\nabla_p \varphi(p,x)$ and $\nabla_x \varphi(p,x)$.
\end{enumerate}

None of the above information is readily available, making a direct implementation of the recommender design \eqref{PGD_ones} impractical. In fact, the users' opinion are not accessible in practice, their dynamics are generally unknown (and so are their sensitivities), and the gradients $\nabla_p \varphi$, $\nabla_x \varphi$ depend on the users' clicking behaviour $g(p,x)$, which is also unknown. Instead, we have access only to the users' clicks $\{c^k\}$ from which we compute their CTRs $y_i^k := \sum_{t=k-T}^k c_i^t/(T+1), \forall \ k\geq T$, for some time window $T$. This can be regarded as an approximation of the expectation of a click, i.e., $\mathbb{E}\big[\mathcal{B}(g(p,x))\big] = g(p,x)$.


To cope with these challenges, we augment the controller \eqref{PGD_ones} with three auxiliary blocks associated with  users' opinions, sensitivity, and clicking behaviour estimates in real-time. Specifically, we structure the design on three levels:
\begin{enumerate}[(i)]
\item Estimation of real-time opinions and users' clicking behaviour via supervised learning;
\item Online sensitivity learning via Kalman filtering; and
\item Gradient estimation via a forward difference method.
\end{enumerate}

The resulting control architecture is illustrated in  Fig.~\ref{fig:blockdiagram_2} and summarized in Algorithm \ref{alg:recommender}.
In the following three sections, we describe each level.


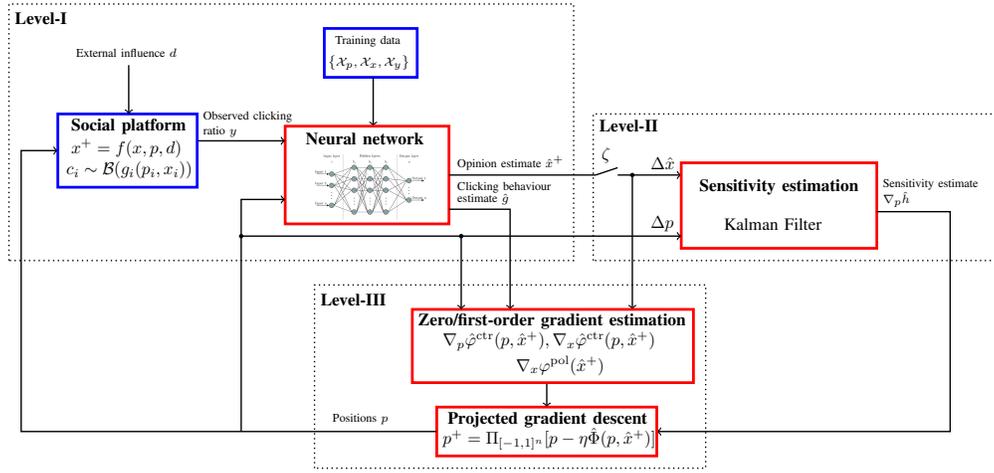
\begin{figure*}
    \begin{center}
    \scalebox{0.65}{\begin{tikzpicture}
    \draw[blue,ultra thick] (0.2,5) rectangle (2.1,6);
        \filldraw[black] (0.3,5.75) circle (0pt) node[anchor=west]{\scriptsize {Training data}};
        \filldraw[black] (0.15,5.3) circle (0pt) node[anchor=west]{\scriptsize $\big\{\mathcal{X}_p,\mathcal{X}_x,\mathcal{X}_y\big\} $};
        \draw[thick,->](1.2,5) -| (1.2,4);
         \draw[red,ultra thick] (-0.6,2) rectangle (2.75,4);
        \filldraw[black] (-0.1,2.75) circle (0pt) node[anchor=west]{\includegraphics[scale=0.1]{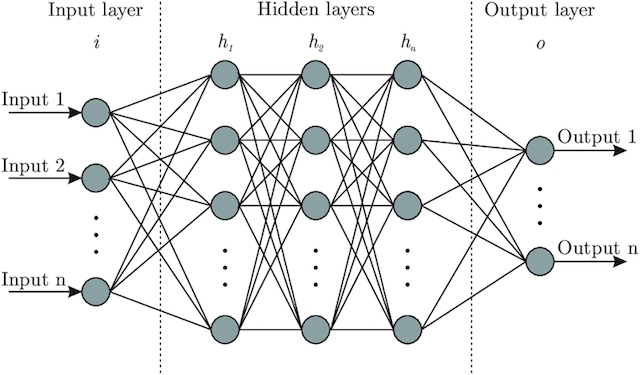}};
        \filldraw[black] (-0.3,3.75) circle (0pt) node[anchor=west]{\textbf{Neural network}};
        \draw[thick,-](2.75,3) -- (5.75,3);
        \filldraw[black] (2.8,3.25) circle (0pt) node[anchor=west]{\scriptsize Opinion estimate $\hat{x}^+$};
        \draw[thick,->](2.75,2.3) -- (4,2.3) -| (4,0.25);
        \filldraw[black] (2.8,2.75) circle (0pt) node[anchor=west]{\scriptsize Clicking behaviour};
        \filldraw[black] (2.8,2.5) circle (0pt) node[anchor=west]{\scriptsize estimate $
        \hat{g}$};
        \draw[thick,->](-2.4,3.7) -- (-0.6,3.7);
        \draw[thick,-] (5.75,3) -- (6.2,3.15);
        \filldraw[black] (5.75,3.4) circle (0pt) node[anchor=west]{$\zeta$};
        \filldraw[black] (6.75,3.25) circle (0pt) node[anchor=west]{$\Delta \hat{x}$};
        \filldraw[black] (6.75,2) circle (0pt) node[anchor=west]{$\Delta p$};
        \draw[thick,->] (6.2,3) -- (7.5,3);
        \draw[red,ultra thick] (7.5,1.5) rectangle (11.5,3.25);
        \filldraw[black] (7.75,2.75) circle (0pt) node[anchor=west]{ \textbf{Sensitivity estimation}};
        \filldraw[black] (8.25,2) circle (0pt) node[anchor=west]{Kalman Filter};
        \draw[thick,->] (11.5,2.25) -- (13,2.25) -- (13,-2.25) -- (7,-2.25);
        \filldraw[black] (11.5,2.5) circle (0pt) node[anchor=west]{\scriptsize $\nabla_p \hat{h}$};
        \filldraw[black] (11.5,2.8) circle (0pt) node[anchor=west]{\scriptsize Sensitivity estimate};
        \filldraw[black] (2.5,-0.4) circle (0pt) node[anchor=west]{$\nabla_p \hat{\varphi}^{\mathrm{ctr}}(p,\hat{x}^+), \nabla_x \hat{\varphi}^{\mathrm{ctr}}(p,\hat{x}^+)$ };
        \filldraw[black] (4,-0.9) circle (0pt) node[anchor=west]{ $\nabla_x {\varphi}^{\mathrm{pol}}(\hat{x}^+)$};
        \draw[thick,->] (6.5,3) -| (6.5,0.25);
        \filldraw[black] (2,0) circle (0pt) node[anchor=west]{\textbf{Zero/first-order gradient estimation}};
        \draw[thick,->] (4.75,-1.25) -| (4.75,-1.75);
        \draw[red,ultra thick] (2,-1.25) rectangle (7.75,0.25);
        \filldraw[black] (2.6,-2) circle (0pt) node[anchor=west]{\textbf{Projected gradient descent}};
        \filldraw[black] (2.5,-2.4) circle (0pt) node[anchor=west]{ $p^+ = \Pi_{[-1,1]^n}[p - \eta\hat{\Phi}(p,\hat{x}^+)]$};
        \draw[red,ultra thick] (2.5,-2.75) rectangle (7,-1.75);
        \draw[thick,->] (2.5,-2.25) -- (-5.5,-2.25) -| (-6,3.5) -- (-5.25,3.5);
         \filldraw[black] (0.25,-2) circle (0pt) node[anchor=west]{ \scriptsize Positions $p$};
         \draw[thick,->] (-1.5,-2.25) -| (-1.5,2.5) -- (-0.6,2.5);
         \draw[thick,->] (-1.5,1.75) -- (7.5,1.75);
         \draw[thick,->] (3,1.75) -| (3,0.25);
         \filldraw[black] (6.5,3) circle (1pt) node[anchor=west]{};
         \filldraw[black] (-1.5,1.75) circle (1pt) node[anchor=west]{};
         \filldraw[black] (3,1.75) circle (1pt) node[anchor=west]{};
         
         \filldraw[black] (-5.1,4) circle (0pt) node[anchor=west]{\textbf{Social platform}};
         \filldraw[black] (-5.1,3.55) circle (0pt) node[anchor=west]{$x^+ = f(x,p,d)$};
         \filldraw[black] (-5.2,3.1) circle (0pt) node[anchor=west]{ $c_i \sim \mathcal{B}(g_i(p_i,x_i))$};
         \draw[blue,ultra thick] (-5.25,2.75) rectangle (-2.4,4.25);
         \filldraw[black] (-2.4,4.2) circle (0pt) node[anchor=west]{\scriptsize Observed clicking};
         \filldraw[black] (-2.4,3.9) circle (0pt) node[anchor=west]{\scriptsize ratio $y$};
         \draw[thick,->] (-3.8,5.25)  -| (-3.8,4.25);
         \filldraw[black] (-5,5.5) circle (0pt) node[anchor=west]{\scriptsize External influence $d$};

         \draw[dotted,thick,-] (-6.25,1.25) -| (-6.25,6.5) -- (5.3,6.5)-| (5.3,1.25) -- (-6.25,1.25);
         \filldraw[black] (-6.25,6.2) circle (0pt) node[anchor=west]{\textbf{Level-I}};
         \draw[dotted,thick,-] (5.7,1.25) -| (5.7,4.25) -- (14,4.25)-| (14,1.25) -- (5.7,1.25);
         \filldraw[black] (5.7,4) circle (0pt) node[anchor=west]{\textbf{Level-II}};
         \draw[dotted,thick,-] (0,-3) -| (0,0.75) -- (8,0.75)-| (8,-3) -- (0,-3);
         \filldraw[black] (0,0.45) circle (0pt) node[anchor=west]{\textbf{Level-III}};  
\end{tikzpicture}}
\end{center}
    \caption{Block diagram of the proposed recommender system design with the three levels. Level I: Opinion and clicking behaviour estimation, Level II: Sensitivity estimation, Level III: Optimization.}
    \label{fig:blockdiagram_2}
\end{figure*}
\subsection{Level I: Opinions \& Clicking Behaviour Estimation}
The opinions and clicking behaviours are estimated using an Artificial Neural Network (ANN). 

We estimate the opinions by relying on implicit feedback from the users, i.e., the CTR $y \approx \mathbb{E}\big[\mathcal{B}(g(p,x))\big]$ of the users on recommendations with position $p$. Thus, we seek a map to extract the opinions from the clicks and the positions, $x = \beta(y,p)$.  
To bound the resulting opinions and clicking behaviour estimation error from the ANN, we work under the following regularity assumptions \cite{Cothren,SarahDean_CEP}.


\begin{assumption}[Opinion \& Clicking behaviour maps] We assume that
\label{ANNassumption}
\begin{enumerate}[i)]
\item There exists a continuous map $\beta: [0,1]^n\times [-1,1]^n \to [-1,1]^n$ is such that $\beta(y,p) = x + \theta(x)$  with $\lVert\theta(x)\rVert \leq \theta_x $, for all $x \in \mathcal X \subseteq [-1,1]^n$, and $\theta_x < \infty$; moreover, the image $\beta(\mathcal{Y},\mathcal{P})$ is compact for any $\mathcal{Y} \subseteq [0,1]^n,\mathcal{P} \subseteq [-1,1]^n$.
\item There exists a constant $\alpha_{y}<\infty$ such that the composite function takes the form $g(p,\beta(y,p)) = y + \nabla_x g(p,x)^\top \theta(x) + \alpha(y)$, with $\lVert\alpha(y)\rVert \leq \alpha_{y}$, for all $y\in \mathcal{Y}$.    
{
  \hfill $  \square$
    }
\end{enumerate}
\end{assumption}
Assumption \ref{ANNassumption} i) guarantees the map $\beta(y,p)$ to provide an estimate on users' opinions upper bounded by $\theta_x$. Similarly, Assumption \ref{ANNassumption} ii), guarantees the (opinion) map $\beta(y,p)$ to (implicitly) provide an estimate on clicking behavior upper bounded by $M_x\theta_x + \alpha_y$, with $M_x$ resulting from Assumption \ref{measurementmodelassumpt}. We will qualitatively describe how these upper bound errors can be tightened in Section \ref{sec:cl_convergence}.  


In the remainder of this section, we denote by $\mathcal{X}_p , \mathcal{X}_x \in [-1,1]^{n \times m} , \mathcal{X}_{y} \in [0,1]^{n \times m}$ the data sets of the positions, steady-state opinions, and the CTR used for training of the ANN. The process to acquire the training data is explained in detail in Appendix \ref{trainingdata}.

\subsubsection{Opinion Estimation}
\label{op_est}
Here, we describe the training process to estimate ${\beta}$ in Assumption \ref{ANNassumption} using an ANN.

For each user $i\in[n]$, the input and output training sets are
\begin{equation*}
    \begin{bmatrix}
    \mathcal{X}_p\\
    \mathcal{X}_{y,i}
\end{bmatrix} \in \begin{Bmatrix}
    [-1,1]^{n \times m}\\
    [0,1]^{1 \times m}
\end{Bmatrix}, \mathcal{X}_{x,i} \in[-1,1]^{1 \times m},
\end{equation*}
respectively. Given that users interact with each other, the recommendations to users connected to user \(i\) will indirectly influence user \(i\)'s opinion \(x_i\). However, \(x_i\) is not directly affected by the CTR of other users. To model this, we design the input layer with \(n+1\) neurons for each user: \(n\) neurons associated with \(p\) and one neuron associated with \(y_i\). The output layer comprises a single neuron to represent \(x_i\). Additionally, we incorporate an intermediate layer with \(n+2\) neurons. The intermediate layer employs a hyperbolic tangent activation function \(\tanh(x): \mathbb{R} \rightarrow [-1,1]\), while the output layer uses an identity activation function.


Next, we provide an upper bound for the opinion estimation error $e_x:= h(p,d)-\hat{\beta}(y,p)$, where $\hat{\beta}$ is the opinion estimate map learned by the neural network. 
\begin{lemma}[Opinion estimation error]\label{est_opinion}
Under Assumption~\ref{ANNassumption}, the opinion estimation error is upper bounded as $\lVert e_x\rVert \leq\epsilon_x$ with
\begin{align}
    \epsilon_x :=\sqrt{n}\Big[&3\underset{y\in\mathcal{X}_y,p\in\mathcal{X}_p}{\mathrm{sup}}\lVert\beta(y,p) - \hat{\beta}(y,p)\rVert_\infty + 2\underset{i\in[n]}{\mathrm{sup}}\omega_{\beta_i}(\gamma_x) + \nonumber\\
    &\gamma_x\underset{i\in[n]}{\mathrm{sup}}|v_{0,i}|\Big]+\theta_x, \label{upperbound_ANN_x}
\end{align}
where $\underset{y\in\mathcal{X}_y,p\in\mathcal{X}_p}{\mathrm{sup}}\lVert\beta(y,p) - \hat{\beta}(y,p)\rVert_\infty$ represents the maximum opinion estimation error during training of the ANN, $\omega_{\beta_i}(\gamma_x)$ is the minimum modulus of continuity of $\beta_i$ on the training set as in Definition \ref{modcontinuity}, and $v_{0,i}$ denotes the bias from the hidden layer to the output layer of the ANN on user~$i$.
\end{lemma}
\begin{proof}
Refer to Appendix \ref{est_opinionproof}.    
\end{proof}


\subsubsection{Clicking Behaviour Estimation}\label{cb_est}
\begin{tac}
For each user $i\in[n]$, the input and output training sets are $\begin{bmatrix}
    \mathcal{X}_{p,i}\\
    \mathcal{X}_{x,i}
\end{bmatrix} \in 
[-1,1]^{2 \times m}$, $\mathcal{X}_{y,i} \in[0,1]^{1 \times m}$, respectively. We then construct an ANN similarly as for opinion estimation\footnote{We omit the details due to space limitations and refer the interested reader to \cite[]{} for a detail discussion.}.
\end{tac} 
\begin{arxiv}
To estimate the clicking behaviour, we proceed similarly as for opinion estimation. 
For each user $i\in[n]$, the input and output training sets are
\begin{equation*}
    \begin{bmatrix}
    \mathcal{X}_{p,i}\\
    \mathcal{X}_{x,i}
\end{bmatrix} \in [-1,1]^{2 \times m}, \mathcal{X}_{y,i} \in [0,1]^{1 \times m},
\end{equation*}
respectively. Given the recommender system's personalization of positions to the user's interests, the $i$-th user's CTR $y_i$ only depends on their steady-state opinion $x_i$ and the provided position $p_i$. Hence, we have two input neurons ($x_i,p_i$) and one neuron in the output layer ($y_i$). Further, we use $3$ intermediate layers with $5$ neurons. The hyperbolic tangent and the identity are used as activation functions at the intermediate layers and the output layer, respectively.
\end{arxiv}

We now derive an upper bound on the resulting clicking behaviour estimation error $e_y:= \hat{g}(p,\hat{x}) - g(p,h(p,d))$, where $\hat{x}$ and $\hat{g}$ represents the opinion and clicking behaviour learned by the ANN. 
\begin{lemma}[Clicking behaviour estimation error]
\label{est_clickingbehaviour}
Under Assumptions \ref{measurementmodelassumpt}--\ref{ANNassumption}, the clicking behaviour estimation error is upper-bounded as $\lVert e_y \rVert \leq \epsilon_g$, with
\begin{align} 
    \epsilon_g \leq \sqrt{n}\Big[&3\underset{p\in\mathcal{X}_p,x\in\mathcal{X}_x}{\mathrm{sup}}\lVert g(p,x) - \hat{g}(p,x)\rVert_\infty + 2\underset{i\in[n]}{\mathrm{sup}}\omega_{g_i}(\gamma_y) + \nonumber\\
    &\gamma_y\underset{i\in[n]}{\mathrm{sup}}|w_{0,i}|\Big]+M_x\epsilon_x + \alpha_y,\label{upperbound_ANN_c}
\end{align}
where $\underset{p\in\mathcal{X}_p,x\in\mathcal{X}_x}{\mathrm{sup}}\lVert g(p,x) - \hat{g}(p,x)\rVert_\infty$ is the maximum estimation error during ANN training, $\omega_{g_i}(\gamma_y)$ denotes the minimum modulus of continuity  of $g_i$ on the training set as in Definition~\ref{modcontinuity}, and $w_{0,i}$ represents the estimated bias weight from the hidden layer to the output layer of the ANN on user~$i$.
\end{lemma}
\begin{tac}
\begin{proof}
The proof is similar to that of Lemma \ref{est_opinion} and hence is omitted due to space limitation. We refer the interested reader to the extended version of this paper \cite{}.
\end{proof}
\end{tac}

\begin{arxiv}
\begin{proof}
Refer to Appendix \ref{est_clickingbehaviourproof}.
\end{proof}
\end{arxiv}
The values $\gamma_x,\gamma_y$ and the modulus of continuity are properties of the training set and can be determined before training the ANN. The maximum training error term in Lemmas \ref{est_opinion}, \ref{est_clickingbehaviour} can be ascertained after training the ANN. The values $v_{0,i},w_{0,i}$ are the bias weights from the penultimate to the final layer in the ANN, which are known after training the ANN.

\subsection{Level II: Online Sensitivity Learning}
To estimate the sensitivity $\nabla_p h(p,d)$ in \eqref{eq:grad}, we develop  a Kalman filter, similar to \cite{picallo}. Given that opinion dynamics and recommendations evolve over time scales that are comparable in magnitude, in order to make our OFO scheme \cite{hauswirth1} applicable, we introduce an auxiliary mechanism  that guarantees time-scale separation  between the opinion dynamics and the recommender system updates, see Remark \ref{rem:stationarity}.

Note that the input-output sensitivity brings information about interpersonal influence among agents, in fact, if $\nabla_p h_{ij}(p,d)\neq 0$ it means that the position provided to agent $j$ influences agent $i$'s opinion, hence $j$ and $i$ are connected over the platform. Thus, by estimating the sensitivity we are able to perform network-aware recommendations. We will see later in this manuscript how this turns out to be beneficial in attenuating polarization.

\subsubsection{Kalman Filter for Sensitivity Learning}
We denote by $\ell \in \mathbb{R}^{n^2}$ the vectorized sensitivity $\nabla_p h(p,d) \in \mathbb{R}^{n \times n}$, namely, $\ell := \textrm{vec}(\nabla_p h(p,d))$, and model the sensitivity dynamics as the following random walk:
\begin{equation} \label{processmodelKF}
    \ell^k = \ell^{k-1} + w^{k-1}, 
\end{equation} 
where, $w^k\sim\mathcal{N}(0_{n^2},Q^k)$ is the process noise, with $Q^k$ as its corresponding covariance matrix. The measurement model is 
\begin{equation}\label{measurementmodelKF}
    \Delta x_{\rm ss}^{k+1,k}  = \Delta \Tilde{p}^{k,k-1} \ell^k + v^k,
\end{equation}
where $\Delta x_{\rm ss}^{k+1,k}:=h(p^k,d) - h(p^{k-1},d)$ is the change in steady-state opinions for a corresponding change in positions $\Delta p^{k,k-1} := p^k-p^{k-1}$, and $\Delta \Tilde{p}^{l,m}:= \Delta (p^{l,m})^\top \otimes I_n \in \mathbb{R}^{n \times n^2}$, where $\otimes$ indicates the Kronecker product. The measurement noise is described by $v^k\sim\mathcal{N}(0_n,R^k)$, where $R^k$ is its corresponding covariance matrix. This noise accounts for the contribution of the external influence $d$ to the change of opinions $\Delta x$. 
\begin{remark}
The measurement model \eqref{measurementmodelKF} is obtained from a first-order Taylor  approximation of the steady-state map around $p^{k-1}$, i.e., $h(p^k,d) \approx h(p^{k-1},d) + \nabla_p h(p^{k-1},d)^\top(p^k - p^{k-1})$. With this intuition, it can be seen that the measurement model is a linear time-variant approximation of the temporal evolution of the sensitivity $\nabla_p h(p^k,d)$. Note that, the steady-state $h(p^k,d)$ is reached only asymptotically. Thus, it is common practice \cite{picallo} to use instead $\Delta x^{k+1,k}_{ss} = x^{k+1} - x^k$ as the measurement in \eqref{measurementmodelKF}. 
{\hfill $\square$}
\end{remark}




The posterior update of sensitivity estimates $\hat \ell^k$ and covariance $\Sigma^k$  are given by
 \begin{align}\label{KFposterior_state}
  \hat{\ell}^k &= \hat{\ell}^{k-1} + \zeta^k K^{k-1}\big(\Delta \hat{x}^{k+1,\tau_i+1} - \Delta \Tilde{p}^{k,\tau_i}\hat{\ell}^{k-1}\big) \notag \\ 
    \Sigma^k &= \Sigma^{k-1} + \zeta^k\big(Q^k - K^{k-1}\Delta \Tilde{p}^{k,\tau_i}\Sigma^{k-1}\big),
\end{align}
where $\zeta$ enforces an auxiliary trigger mechanism, with
\begin{equation*}
    \zeta^k = \Bigg\{\begin{array}{ll}
        1, & k = iT \\
        0, & \rm otherwise 
    \end{array}.
\end{equation*}
The trigger mechanism is introduced for the recommender system to update positions close to steady state, see Remark \ref{rem:stationarity}, and to collect a sufficient number of clicks to ensure a proper estimate of the CTR needed for opinion estimation.
 The trigger mechanism is turned on at integer multiples $i\in \mathbb{N} $ of the time period $T$. We define \(\mathcal{T}\) as the set of all trigger time instances, with \(\tau_i \in \mathcal{T}\) representing the most recent trigger time instance prior to \(k\). 
Note that the update \eqref{KFposterior_state} uses the opinion estimate $\hat{x}$ introduced in \S\ref{op_est} rather than the real users' opinion, with $\hat{x}^{k+1} = \hat{\beta}(y^k,p^k)$, where $y^k := \sum_{t=\tau_i}^k c^t/(k-\tau_i+1)$ is the CTR since the most recent trigger time instant. Finally, the Kalman filter gain $K^k$ in \eqref{KFposterior_state} is given by
\begin{equation}\label{KFgain}
    K^k = \Sigma^k (\Delta\Tilde{p}^{k,\tau_i})^\top(R^k + \Delta \Tilde{p}^{k,\tau_i}\Sigma^k(\Delta \Tilde{p}^{k,\tau_i})^\top)^{-1},
\end{equation}
with $\Delta \tilde{p}^{k,\tau_i}$ and $R^k$ as in the measurement model \eqref{measurementmodelKF}. Next, we postulate some regularity assumptions for the process and measurement models.

\begin{assumption}[Gaussian noise]\label{IIDassumption}
    The process and measurement noise $w,v$ are white Gaussian. Moreover, the steady-state opinion estimation error $e_x=h(p,d)-\hat{\beta}(y,p)$ is uncorrelated with $w$ and $v$. {
  \hfill $  \square$
    }
\end{assumption}

The assumption of white noise for the process model is standard in the context of sensitivity learning in feedback optimization \cite{picallo}. Intuitively, the process perturbation determines the degree of trust one puts on the sensitivity estimates.  It is also reasonable to assume that the external influence is uncorrelated among users in certain cases, for example the extended FJ model in \eqref{eq:FJdyn} where $d= x_0$. We also note that since $\hat{x}$ is estimated using the ANN, the steady-state estimation error $e_x$ is not correlated with the process or measurement noises in \eqref{processmodelKF}, \eqref{measurementmodelKF}. Under Assumption \ref{IIDassumption}, the covariances simplify as $Q^k = (\sigma_q^k)^2I_{n^2}$ and $R^k = (\sigma_r^k)^2I_n$, for some $\sigma_q$, $\sigma_r >0$. The tuning of $\sigma_q,\sigma_r$, is based on heuristics and is described in Appendix~\ref{HPtuning}.

To ensure that the sensitivity is correctly inferred, we must guarantee that the input positions $\Delta p$ are persistently exciting (see Definition \ref{POEhenkel}). This translates to $\Delta p$ being sufficiently variable in order to explore the underlying user opinion dynamics.
To do so, we work under the following assumption.
\begin{assumption}[Persistency of excitation] \label{POEassumption}
    The inputs $\Delta p$ are persistently exciting (see Definition \ref{POEhenkel}) and the trigger mechanism in \eqref{KFposterior_state} is activated at least $S+1$ times, namely $|\mathcal{T}|\geq S+1$.
    {
  \hfill $  \square$
    }
\end{assumption}
In practice, this assumption can be satisfied by introducing a dither signal in the update rule \eqref{PGD_ones}.


\subsection{Level III: Gradient Estimation \& Optimization}
We describe the gradient estimation method for the engagement cost function in \eqref{eq:obj_rec}. We remark that this algorithmic step is needed as the clicking behaviour is unknown. We then describe the OFO update rule from \eqref{PGD_ones}, incorporating gradient, sensitivity and state estimation. 

\subsubsection{Gradient Estimation}
To estimate the gradient of the engagement maximization cost $\varphi^{\rm ctr}$, we use the \textit{finite forward difference} method   \cite{katya}
\begin{equation}\label{gradapprox_x}
    \nabla_x\hat{\varphi}^{\rm ctr}_i(p,x) = \frac{\hat{\varphi}^{\rm ctr}(p,x + \mu e_i)-\hat{\varphi}^{\rm ctr}(p,x)}{\mu}, 
\end{equation}
\begin{equation}\label{gradapprox_p}
     \nabla_p\hat{\varphi}^{\rm ctr}_i(p,x) = \frac{\hat{\varphi}^{\rm ctr}(p + \mu e_i,x)-\hat{\varphi}^{\rm ctr}(p,x)}{\mu},
\end{equation}
where $\nabla_j\hat{\varphi}^{\rm ctr}_i, j\in\{p,x\}$ denotes the $i$-th entry of the gradient, $e_i \in \mathbb{R}^n$ refers to the $i$-th vector of the canonical basis of $\mathbb{R}^n$, and $\mu$ is a smoothing parameter. Note that in the two-point gradient estimation, we do not actuate the opinion dynamics twice but rather use the engagement cost function estimate $\hat{\varphi}^{\rm ctr}(p,x) = -1_n^\top \hat{g}(p,x)$ as described in \S \ref{cb_est}.

The smoothing parameter $\mu$ is chosen small enough so that the gradient estimate provides a good approximation of the true value. However, choosing $\mu$ in \eqref{gradapprox_x}--\eqref{gradapprox_p} too small can lead to numerical instability, as formalized in the following.
\begin{lemma}[Gradient estimation error]
\label{closeness}
    Under Assumptions \ref{measurementmodelassumpt}--\ref{ANNassumption}, the gradient estimation error is upper bounded as 
    \begin{equation*}\label{closeness_eqn}
        \lVert\nabla_j \hat{\varphi}^{\rm ctr}-\nabla_j \varphi^{\rm ctr}\rVert \leq \frac{1}{2}L_x\mu + 2\frac{\sqrt{n}\epsilon_g}{\mu}, \quad j\in\{x,p\}.
    \end{equation*}
    The upper bound attains its lowest value $2n^{3/4}\sqrt{\epsilon_g L_j}$ at $\mu^* = 2n^{1/4}\sqrt{{\epsilon_g}/{L_j}}$, with $ j\in \{x,p\}$.
\end{lemma}
\begin{proof}
    Refer to Appendix \ref{closenessproof}.
\end{proof}

\subsubsection{Projected Gradient Descent Algorithm}
We now present the projected-gradient update in \eqref{PGD_ones} augmented with sensitivity, state, and gradient estimations. The augmented update reads compactly as
\begin{equation}\label{PGD_updated}
    p^{k+1} = \Pi_{[-1,1]^n}\Big[p^k - \zeta^k\eta \,\hat{\Phi}^k \Big],
\end{equation}
where the gradient surrogate $\hat{\Phi}$ is constructed by combining sensitivity, opinion, and gradient estimates as
\begin{align}
 \hat{\Phi}^k =& \nabla_p \hat{\varphi}^{\rm ctr}(p^k,\hat{x}^{k+1}) + (\hat{H}^k)^\top\nabla_x \hat{\varphi}^{\rm ctr}(p^k,\hat{x}^{k+1}) \nonumber\\
 &+\gamma(\hat{H}^k)^\top\nabla_x\varphi^{\rm pol}(\hat{x}^{k+1}) - w_{\rm pe}^k/\eta, \label{compositegradientestimate}
\end{align}
where $\hat{H}^k := \nabla_p \hat{h}(p^k,d)$ is the sensitivity estimate at time $k$. The additional term $w_{\rm pe}^k \sim \mathcal{N}(0_n,\sigma_{\rm pe}^2I_n)$ is a dither signal that ensures persistency of excitation of the inputs (Assumption \ref{POEassumption}).

In Algorithm \ref{alg:recommender},  we summarize the pseudo-code of the proposed recommender system design. In the offline phase, training of the neural network is carried out for opinion and {clicking behaviour} estimation. In the online phase, a new sensitivity estimate is generated and new positions are provided to the users periodically, every $T$ time steps.

\begin{algorithm}
\caption{$\textrm{Recommender}[T,n]$} \label{alg:recommender}
    \begin{algorithmic}
    \STATE \textbf{Initialization}
    \STATE Collect data using Algorithm \ref{alg:training}
    \STATE Build opinion and {clicking behaviour} estimators ($\hat{\beta},\hat{g}$) 
    \STATE \textbf{Optimization phase}
    \FOR{$k\geq 0$}
        \STATE Collect clicks $c_i^k\sim\mathcal{B}\big(g_i(p_i^k,x_i^k)\big)$ from users
        \STATE CTR $y^k \leftarrow \frac{\sum_{t=\tau_i}^k c^t}{k-\tau_i+1}, \tau_i=(i-1)T<k$
        \STATE Estimate opinions $\hat{x}_i^{k+1} \leftarrow \hat{\beta}_i(y_i^k,p^k)$
        \IF{$\zeta^k=1$}
            \STATE $\mathcal{T} \leftarrow \textrm{append}[k]$
            \STATE Estimate sensitivity $\hat{H}^k$ using \eqref{KFposterior_state}--\eqref{KFgain}
            \STATE Estimate gradient using \eqref{gradapprox_x}--\eqref{gradapprox_p}
            \STATE Update positions $p^{k+1}$ using \eqref{PGD_updated}
        \ELSE
            \STATE $\hat{H}^k \leftarrow \hat{H}^{k-1}$
            \STATE $p^{k+1} \leftarrow p^k$
        \ENDIF
    \ENDFOR
    \end{algorithmic}
\end{algorithm}

\subsection{Closed-loop Convergence Guarantees}\label{sec:cl_convergence}
Here, we establish convergence guarantees for the sensitivity estimation process \eqref{KFposterior_state} and for the closed-loop interconnection between the opinions and recommendations.
Before stating the formal results, we make some important observations regarding our opinion estimation technique and its implication on the closed-loop convergence analysis.

In our problem setup, the measurement information at each time instant (i.e., a single click/no click data), is not rich enough to produce accurate estimates for the users' real-time opinion.
As shown in Algorithm~\ref{alg:recommender}, the opinion estimation is instead carried out using the CTR recorded on multiple recommendations with a constant position since the previous trigger instant. In light of the exponential stability of the opinion dynamics (i.e., Assumption \ref{ass:ODasm}), the resulting estimate at the end of each trigger period is, in practice, close to an estimate of the steady-state opinion $h(p^k,d)$. This approach is therefore substantially different from real-time state estimation \cite{Cothren}. Thus, as we directly estimate steady-state opinions at each trigger instant, we can greatly simplify our closed-loop stability analysis by treating the opinion dynamics as a static mapping.
Also note that as the trigger period $T$ increases, the resulting CTR offers a better approximation of the CTR in expectation $\mathbb{E}\big[g(p,x)\big]$. The modeling errors $\theta_x,\alpha_y$ in Lemma~\ref{est_opinion},~\ref{est_clickingbehaviour} exist due to a finite time horizon $T$ and reduces as $T$ increases. 

Now, we state the main convergence result with respect to the sensitivity estimation error.
\begin{theorem}[Sensitivity estimate convergence]
\label{sensitivityconvergence}
    Under Assumptions \ref{ass:ODasm}, \ref{IIDassumption}, and \ref{POEassumption}, the sensitivity estimation error  $e^k:=\ell^k-\hat{\ell}^k$ has bias and variance bounded in norm, i.e.,  there exist positive constants $c_1,c_2,C_f < \infty$ and $\xi\in(0,1)$ such that
\begin{align}
\lVert \mathbb{E}[e^k]\rVert &\leq 2K_m\epsilon_x/(1-c_1\xi^{c_2T}) + (c_1\xi^{c_2T})^{|\mathcal{T}|}\lVert \mathbb{E}[e^0]\rVert \nonumber\\
\mathbb{E}\Big[\lVert e^k\rVert^2\Big] &\leq C_f + (c_1\xi^{c_2T})^{2|\mathcal{T}|}\mathbb{E}\Big[\lVert e^0\rVert^2\Big] \nonumber
\end{align}
    where $C_f = \frac{1}{1-(c_1\xi^{c_2T})^2}\big[ (T-1)\overline{\sigma}_q^2 + K_m^2\big(\overline{\sigma}_r^2 + 2\epsilon_x^2\big)\big]$ with $\overline{\sigma}_q = \underset{k\in\mathbb{N}_0}{\mathrm{sup}}{\sigma}_q^k$, $\overline{\sigma}_r = \underset{k\in\mathbb{N}_0}{\mathrm{sup}}{\sigma}_r^k$, $K_m = \underset{k\in\mathbb{N}_0}{\mathrm{sup}} \lVert K^k\rVert$, $c_1\xi^{c_2T}<1$.
\end{theorem}
\begin{proof}
    Refer to Appendix \ref{sensitivityconvergenceproof}.
\end{proof}

Note that both bias and variance upper bounds depend on the opinion estimation error $\epsilon_x$. Further, note that increasing the sampling period $T$ reduces the bias and variance error through the term $1-c_1\xi^{c_2T}$. This is expected as increasing $T$ guarantees greater time-scale separation between opinion dynamics and recommendations. Finally, the upper bound on the sensitivity error variance is proportional to the noise variance of the opinions through the term $\overline{\sigma}_r^2$. 

To quantify performance on the recommendations, we use the \textit{fixed-point residual mapping} \cite[Eq. (5)]{gradientmapping}
\begin{equation}\label{gradientmapping}
    \mathcal{G}(p) := \frac{1}{\eta}\Big(p - \Pi_{[-1,1]^n} \big[p - \eta\Phi(p,h(p,d)) \big]\Big).
\end{equation}
The fixed-point residual mapping is zero at any critical point of \eqref{opt}, and is a common metric to quantify convergence of iterative algorithms in non-convex regimes \cite{nesterov}.

The following theorem shows the convergence of the closed-loop system, combining Levels I, II, and III in Figure \ref{fig:blockdiagram_2}, using the second moment of the fixed-point residual mapping. Since the problem is non-convex, the trajectories of $\{p^k\}_{k=0}^N$ vary across each trial. Thus, it is essential to evaluate $\lVert \mathcal{G}(p^k) \rVert$ with its expectation. Moreover, we compute the average of this metric over all the trigger time instances to reflect the average performance of Algorithm~\ref{alg:recommender} over the time in which recommendations are provided. The theoretical bound is only computed over the trigger time instances since the positions remain constant between two trigger instances, and so does the residual $\lVert \mathcal{G}(p^k) \rVert$.
\begin{theorem}[Closed-loop convergence]
\label{OFOconvergence}
    Let Assumptions \ref{ass:ODasm}--\ref{POEassumption} hold. For any $\eta \in (0,\frac{1}{2L'})$ and $\mu^*= 2n^{1/4}\sqrt{{\epsilon_g}/{L_j}}, j\in\{p,x\}$, the sequence $\{p^k\}_{k \in \mathbb{N}}$ generated by \eqref{PGD_updated} satisfies
    \begin{align}
        \frac{1}{|\mathcal{T}|}\sum_{\underset{l\leq k}{l\in\mathcal{T}}}\mathbb{E}\Big[\lVert\mathcal{G}(p^l)\rVert^2\Big] &\leq K_1, \quad \forall
        k\geq T \label{OFOconvergenceequation}
    \end{align}
where the upper bound $K_1$ is given by
    \begin{align*}
        &K_1=\frac{6}{1-2\eta L'}\Bigg(\underset{\kappa_1}{\underbrace{\frac{\varphi(p^0,h(p^0,d))-\varphi^{*}}{3\eta|\mathcal{T}|}}} +\frac{\sigma_{\rm pe}^2}{\eta^2}+4\Bigg(\underset{\kappa_2}{\underbrace{\gamma^2L^2\epsilon_x^2}} + \\
        &\underset{\kappa_3}{\underbrace{n^{3/2}\epsilon_g\!(L_p\!\!+\!\!L^2L_x)}}\!\Bigg)\!\!\!+\!\!{2(2n\gamma^2\!\!+\!\!M_x^2\!\!+\!\!4n^{3/2}\epsilon_gL_x)}\frac{\sum_{\underset{l\leq k}{l\in\mathcal{T}}}\mathbb{E}[\lVert e^l \rVert^2]}{|\mathcal{T}|} \ \!\!\Bigg),
    \end{align*}
    In there,  $L'=L_p + L^2(L_x + 2)$ is the Lipschitz constant of the gradient $\Phi(\cdot,h(\cdot,d))$, and $\varphi^{*}={\mathrm{inf}}_{p\in[-1,1]^n}\varphi(p,h(p,d))>-n$ is the cost function value at a locally optimal solution. 
\end{theorem}
\begin{proof}
    Refer to Appendix \ref{OFOconvergenceproof}.
\end{proof}
  Note that an upper bound on the minimum of $\lVert \mathcal{G}(p^k) \rVert$ can be inferred from \eqref{OFOconvergenceequation}, through the inequality $\mathbb{E}\Big[\underset{l\in\mathcal{T}}{\mathrm{min}}\lVert\mathcal{G}(p^l)\rVert^2\Big] \leq \mathbb{E}\Big[\frac{1}{|\mathcal{T}|}\sum_{l\in\mathcal{T}}\lVert\mathcal{G}(p^l)\rVert^2\Big] = \frac{1}{|\mathcal{T}|}\mathbb{E}\Big[\sum_{l\in\mathcal{T}}\lVert\mathcal{G}(p^l)\rVert^2\Big]$, see \cite[Remark 5]{zhiyu}. Thus, intuitively, the bound in \eqref{OFOconvergenceequation} describes the maximum possible minimum gap between the trajectory $\{p^k,h(p^k,d)\}_{k=0}^N$ obtained using Algorithm \ref{alg:recommender} and a local minima $(p^*,h(p^*,d))$.  

The bound $K_1$ shows that achieving an accurate solution of \eqref{opt}, i.e., a local minima $(p^*,h(p^*,d))$, is limited by the deviation of the initial cost at $(p^0,h(p^0,d))$ from the optimal one $\varphi^{*}$ (i.e., the term $\kappa_1$), the polarization and engagement gradient estimation errors (i.e., the terms $\kappa_2,\kappa_3$, respectively) including both opinion and clicking behaviour estimation errors, the variance of the dither signal $\sigma_{\rm pe}^2$, and the sensitivity estimation error variance $\mathbb{E}[\lVert e^l \rVert^2]$.   

\begin{corollary}[Asymptotic convergence]
    As $|\mathcal{T}|\rightarrow\infty$, the upper bound $K_1$ in in Theorem \ref{OFOconvergence} approaches $K_2$, given by
    \begin{align*}
        K_2 = \frac{6}{1-2\eta L'}\Bigg(& 2C_f(2n\gamma^2 + M_x^2 + 4n^{3/2}\epsilon_gL_x) + \\
        & 4(\kappa_2 + \kappa_3) + \frac{\sigma_{\rm pe}^2}{\eta^2} \Bigg) 
    \end{align*}
\end{corollary}
\begin{proof}
    Refer to Appendix \ref{OFOconvergenceproof}.
\end{proof}

From $K_2$, we observe that the asymptotic accuracy is limited by a combination of gradient estimation errors, which include opinion and clicking behaviour estimation (i.e., the terms $\kappa_2,\kappa_3$), the sensitivity estimation error, and the variance of the dither signal $\sigma_{\rm pe}^2$. Moreover, the upper-bound of the average of the variance of sensitivity estimation error over the trigger time instances approaches $C_f$, the asymptotic sensitivity error variance. 

\section{Numerical Case Studies}\label{sec:results}
\subsection{Opinion Dynamics and Clicking Behaviour}
We assume the opinions of the users on the platform evolve according to the FJ  model as in Example \ref{ex:FJ}, i.e.,
\begin{equation}\label{FriedkinJohnsen}
        x^{k+1} = (I_n - \Gamma_p - \Gamma_d)Ax^k + \Gamma_p p^k + \Gamma_d d.\notag
\end{equation}

Further, we consider users with two different clicking behaviours:  confirmation bias towards extreme positions, $C_A$, defined as in Example \ref{ex:clickingbehaviour}, and confirmation bias towards any position, $C_B$, wherein the click event is given by
\begin{equation}\label{C_B}
    c_i \sim \mathcal{B}\Big(\frac{1}{2} + \frac{1}{2}e^{-4(x_i-p_i)^2} \Big).
\end{equation}
 
 Note that the choice of these clicking behaviors make the steady-state optimization problem \eqref{opt} non-convex. 

As a proof of concept, we consider a platform with $n=15$ users. We generate the adjacency matrix $A$ with some random entries zero and other entries sampled from $\mathcal{U}[0,1]$\footnote{This ensures that the non-zero weights in the graph are not insignificant.} such that $A\mathbf{1}_n\leq \mathbf{1}_n$. Thus, the resulting graph need not be strongly connected. We generate the initial opinions as $x^0\sim \mathcal{U}[-1,1]^n$. Moreover, we set $\Gamma_d = 0.5 \mathrm{diag}[|x^0|]$. This choice of $\Gamma_d$ indicates that the users who hold initial extreme opinions tend to be more stubborn. The external influence $d$ is modelled as a prejudice term defined as the initial opinion $x^0$. Further, $\Gamma_p$ is drawn from $\mathcal{D}_{[10^{-2},0.5]}^n$ so that all users are affected by the recommendations. To incorporate diversity in clicking behaviours, we randomly assign $8$ users to follow clicking behaviour $C_A$ and the remaining ones follow $C_B$.

\subsection{Experimental Setting and Hyperparameter Selection}

In the offline phase (Algorithm \ref{alg:training} in Appendix \ref{trainingdata}), we train the neural network for opinion and clicking behaviour estimation with $m=375$ training and $125$ testing points, simulation time $N=100$ and trigger period $T=60$, with the clicks being recorded in the duration $k\in[N-T,N]$. We set the trade-off factor $\gamma=1$ in \eqref{eq:obj_rec}, thus giving equal importance to engagement maximization and polarization mitigation.

In the online phase, we choose $\sigma_{\text{pe}} = 0.07$ in \eqref{PGD_updated} for the first $N/5$ time steps to facilitate sensitivity estimation\footnote{Sensitivity estimation can be halted after convergence on its entries.} and network learning, and $\sigma_{\text{pe}} = 0$ afterwards to favor accurate recommendations. 
In the Kalman filter \eqref{KFposterior_state}, we initialize $\hat{\ell}^0=0.5{\rm vec}(I_n)$ and $\Sigma^0=I_{n^2}$. For gradient estimation in \eqref{gradapprox_x}--\eqref{gradapprox_p}, we set the smoothing parameter $\mu=0.1$. While tuning the process noise standard deviation, we set the tuning parameter $M$ in \eqref{sigmaq} as $M=10$. 
 Further, we start with neutral recommendations, i.e., $p^0=0_n$, and consider a periodic trigger with period $T=60$. For the polarization cost $\varphi^{\rm pol}$ in \eqref{polarizationcost}, we set $\epsilon_1 = -0.5$ and $\epsilon_2 = 0.5$. All the simulations are conducted for $N=2.5\cdot10^{4}$ time instances and over $50$ Monte-Carlo trials. 
 Since clicks are stochastic, the sequence of clicks varies across each trial. Moreover, the dither signal $w_{\rm pe}$ sequence varies across each trial, thus leading to different sensitivity estimates.

\subsection{Architecture Comparisons}\label{study1}
In this section, we compare Algorithm 1 against other similar OFO approaches that benefit from more information.
\begin{table}[]
    \begin{center}
    \caption{Methods for comparison}\label{methodscomparisontable}
    \begin{tabular}{p{1.5cm}|p{1.5cm}|p{1.2cm}|p{2.2cm}}
    \hline
        Method & Sensitivity &  Opinions & Clicking behaviour\\
        \hline \hline 
        $M_1$ (Oracle) & \cmark & \cmark & \cmark\\ 
        $M_2$ & \xmark & \cmark & \cmark\\ 
        $M_3$ & \xmark & \xmark & \cmark\\ 
        $M_4$ (Alg. 1) & \xmark & \xmark & \xmark\\ 
        \hline
    \end{tabular}
    \end{center}
\end{table}
In Table \ref{methodscomparisontable}, we summarize all the methods used for comparison and their attributes. In particular, for the oracle (method $M_1$), we employ the standard projected-gradient controller \eqref{PGD_ones}, wherein opinions can be directly measured, and the sensitivity and users' clicking behavior are both known. 

Method $M_2$ uses online sensitivity estimation. Method $M_3$ uses opinion as well as sensitivity estimation. Finally, method $M_4$ corresponds to Algorithm 1, and includes also clicking behaviour and gradient estimation. All the methods' architecture can be extracted from Figure \ref{fig:blockdiagram_2} by replacing or removing some of the blocks.  Method $M_1$ would not require the ANN in Level-1, the Kalman filter in Level-2 and gradient estimation in Level-3. Method $M_2$ would not require the ANN and gradient estimation. Method $M_3$ would need the ANN only for opinion estimation, but does not require gradient estimation. 

Figure \ref{fig:study1_GM} shows that convergence is empirically established in all the methods through the \textit{fixed-point residual mapping} norm $\lVert\mathcal{G}(p^k)\rVert^2$. In particular, for methods that involve sensitivity estimation (i.e., $M_2-M_4$), we observe a higher fluctuation for the first $N/5$ time instances due to the effect of the dither signal $\sigma_{\text{pe}}$. Our proposed method offers slightly inferior performances compared to the others, which is unsurprising as it is completely data-driven. 

In Figure \ref{fig:sensitivityerror}, we illustrate the relative sensitivity estimation error. 
It is expected that method $M_2$ performs better than $M_3,M_4$ as the opinions are known in real time. However, the methods that involve steady-state opinion estimation, $M_3,M_4$, still show comparable performance. This is due to accurate opinion estimation, as can be seen from Figure \ref{fig:opinionestimates}, where we show the opinion estimates for four randomly chosen users, along with their true values. Finally, Figure \ref{fig:gammainfluence}, shows the effect of $\gamma$ in \eqref{eq:obj_rec} on the trade-off between engagement maximization and polarization mitigation. As expected, increasing $\gamma$, turns into a polarization cost reduction. However, the loss in engagement is not significant. This is due to the fact that roughly half of the agents follow clicking behaviour $C_B$, and hence still get engaged by the non-extreme content that promotes polarization mitigation. A more significant engagement loss would be expected if all the users were to be engaged by extreme content only \cite{rossi}.

\begin{figure}
    \centering
    \includegraphics[scale=0.5]{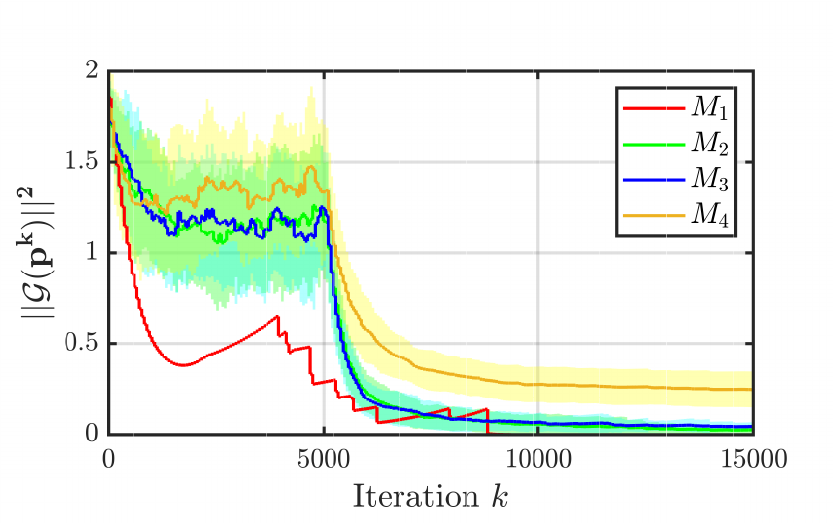}
    \caption{Evolution of the fixed-point residuals $\lVert\mathcal{G}(p^k)\rVert^2$ for the algorithms in Table \ref{methodscomparisontable}. The bold lines represent the mean and the shaded region are the $\pm 1$ standard deviation across the $50$ Monte-Carlo trials.}
    \label{fig:study1_GM}
\end{figure}
\begin{figure}
    \centering
    \includegraphics[scale=0.5]{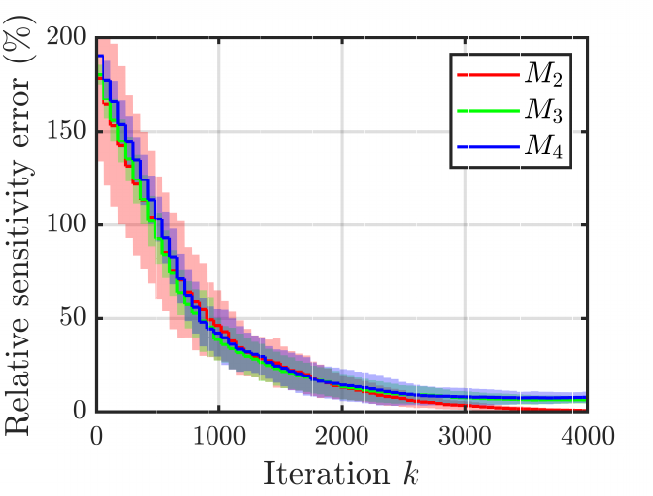}
    \caption{Evolution of the sensitivity estimation error in methods $M_2-M_4$. The bold lines represent the mean and the shaded region are the $\pm 1$ standard deviation across the $50$ Monte-Carlo trials.}
    \label{fig:sensitivityerror}
\end{figure}
\begin{figure}
    \centering
    \includegraphics[scale=0.5]{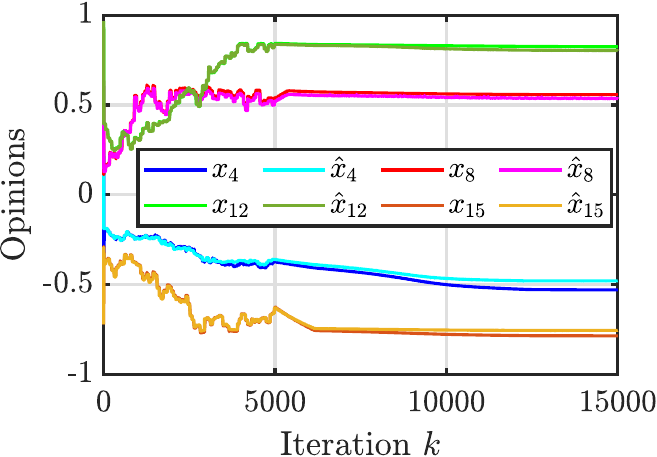}
    \caption{Illustration of opinion estimation in method $M_4$.}
    \label{fig:opinionestimates}
\end{figure}
\begin{figure}
    \centering
    \includegraphics[scale=0.5]{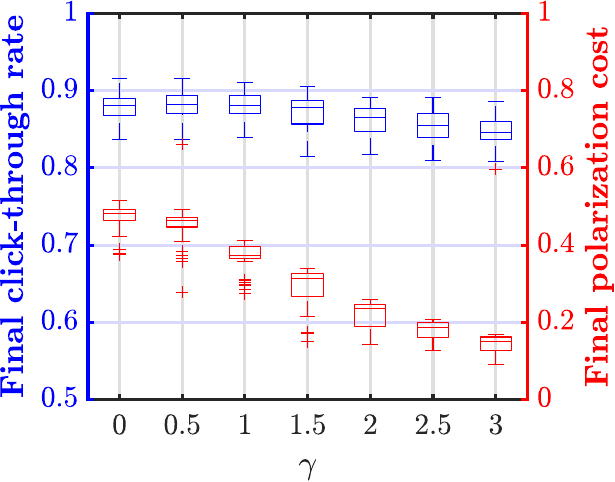}
    \caption{Illustration of the influence of $\gamma$ in \eqref{eq:obj_rec} on final engagement and polarization costs.}
    \label{fig:gammainfluence}
\end{figure}

\begin{figure}
    \subfloat[][]{\includegraphics[scale=0.42]{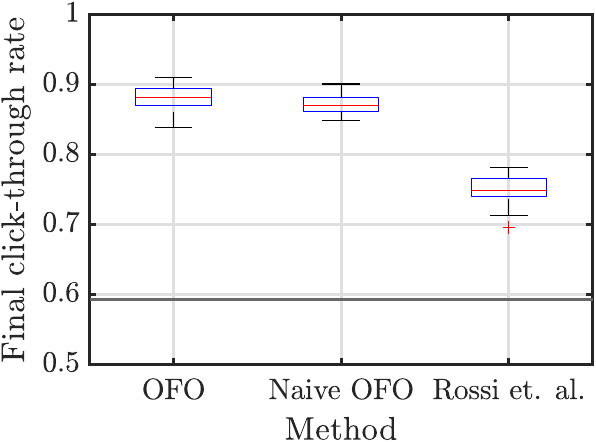}} \ 
    \subfloat[][]{\includegraphics[scale=0.4]{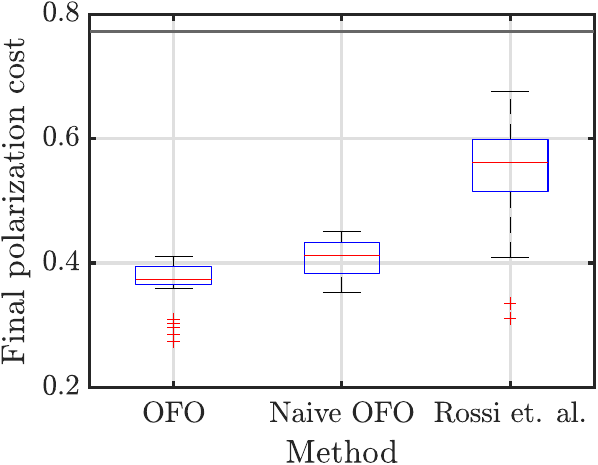}}
    \caption{Steady-state mean CTR (a) and polarization (b) obtained by Algorithm 1 (OFO), its network-agnostic version (naive OFO), and the recommender design in \cite{rossi}. The black lines represent the initial ideal mean CTR (left) and the initial polarization (right).}
    \label{fig:study2_comparison_1}
\end{figure}
\begin{figure}[h]
    \centering
    \includegraphics[scale=0.6]{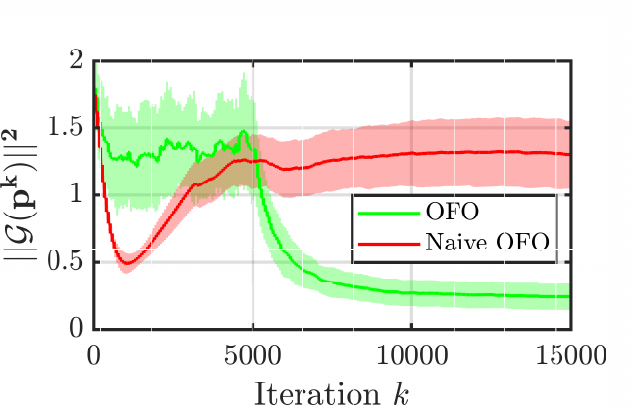}
    \caption{Evolution of the \textit{fixed-point residuals} $\lVert\mathcal{G}(p^k)\rVert^2$ obtained by applying Algorithm 1 and its network-agnostic version. The solid lines represent the mean and the shaded region ($\pm 1$ standard deviation)  the range of changes across the $50$ Monte Carlo trials.}
    \label{fig:study2_GM}
\end{figure}

\subsection{Impact of Network-aware Recommendations}\label{study2}
In this case study, we showcase the impact of providing network-aware recommendations compared to individual, decoupled recommendations. To achieve this, we compare our approach with two network-agnostic methods: (i) the recommender system described in \cite{rossi}, and (ii) a naive version of our OFO-based recommender design. 

The design in \cite{rossi} provides recommendations with exclusively extreme positions (i.e., \( p_i = \pm 1 \))  to each user independently. The position is calculated based on the user's individual clicking history. Specifically, if a user more frequently clicks on \( +1 \) (or \( -1 \)) positions, the system will provide \( +1 \) (or \( -1 \)) positions with high probability. Conversely, if the user shows an equal likelihood of clicking on both \( +1 \) and \( -1 \) positions, the recommendations will be random. It's important to note that the comparison with the method in \cite{rossi} is not entirely equitable in terms of methodology: their approach relies solely on extreme positions and does not target polarization reduction as an objective. However, we include this comparison because, to the best of our knowledge, it is the only model-free, opinion agnostic, systems-theory-based recommender design available.

The second method considered for comparison is a network-agnostic version of our recommender system design summarized in Algorithm \ref{alg:recommender}. We will refer to this method as naive OFO. Specifically, for each user $i\in[n]$, we use a local ANN for opinion estimation that is trained using local positions and acceptance ratio, i.e., $\hat{x}_i = \hat{\beta}_i(y_i,p_i)$ rather than $\hat{\beta}_i(y_i,p)$. Further, we do not carry out sensitivity estimation and instead use a constant diagonal sensitivity $\hat{H} \in \mathcal{D}_{(0,0.5]}^n$ in \eqref{compositegradientestimate}. The rationale behind this design choice is that if the network did not contribute to engagement maximization and/or polarization mitigation, a diagonal \( \hat{H} \) would yield the same performance as our network-aware approach.

 We compare the performance of our recommender system with the other decoupled methods in terms of final mean CTR $J:= \sum_{k=N-T+1}^N \lVert c^k\rVert_1/(nT)$ and final polarization cost $\varphi^{\rm pol}(x^N)$ in Figure \ref{fig:study2_comparison_1} and the fixed point residual mapping in Figure \ref{fig:study2_GM}. In Figure \ref{fig:study2_comparison_1}, we observe that our proposed algorithm achieves a slightly higher ($\sim 1.3\%$) user engagement compared to its naive implementation and results in a lower polarization cost ($\sim 9.3\%$). Thus, network-aware recommendations effectively reduce polarization while maintaining high user engagement. This performance improvement is further illustrated in Figure \ref{fig:study2_GM}, which depicts the temporal evolution of the fixed-point residuals. This result is unsurprising since polarization is a global network metric, whereas engagement maximization can be addressed individually by tailoring recommended content for each user.

Note that the recommender design in \cite{rossi} employs only extreme recommendations, naturally leading to higher polarization costs, as shown in Figure \ref{fig:study2_comparison_1}(b). Furthermore, it results in lower user engagement compared to both our algorithm and the naive method. This is because some users in the network are assigned clicking behavior \(C_B\), for which extreme recommendations do not necessarily increase engagement.



\section{Conclusion}\label{sec:conclusions}

We introduced a novel recommender system algorithm based on online feedback optimization that balances the trade-off between maximizing user engagement (CTR) and minimizing polarization. By integrating auxiliary estimation blocks, our approach overcomes the limitations of existing OFO designs, which require unavailable real-time data, and instead relies solely on click data. We established theoretical guarantees for closed-loop stability and validated our design through numerical simulations using an extended FJ model. Our results demonstrated the impact of network-awareness in reducing polarization while maintaining high engagement.


Future compelling research directions include relaxing the smoothness assumptions for clicking behavior and exploring alternative interest drivers, such as repulsion, beyond confirmation bias. Parameter estimation for clicking behavior could also be refined by using a grey-box model instead of a purely black-box approach. Additionally, investigating the impact of different graph topologies on polarization and developing distributed algorithms for more efficient sensitivity estimation in large-scale networks would be valuable. Lastly, we hope our network-centered approach contributes to the literature on mitigating opinion polarization through algorithmic systems.


\section*{Appendix}
\subsection{Training Data} \label{trainingdata}
The training of the neural network is carried out offline via feed-forward and back-propagation \cite{ANNpaper}.
Algorithm  \ref{alg:training}, provides the pseudo-code to acquire the training data.  

\begin{algorithm}
\caption{[$\mathcal{X}_{p},\mathcal{X}_{x},\mathcal{X}_{y}$] = $\textrm{Training}[N,T,m]$} \label{alg:training}
    \begin{algorithmic}
        \FOR{$j=1$ \textit{to} $m$}
            \STATE $\overline{p} \sim \mathcal{U}[-1,1]^n$
            \STATE $x^{k+1} = f(x^k,\overline{p},d^k)$, $k\in \{0,1,\ldots,N-1\},$
            \STATE $c_i^k = \mathcal{B}(g_i(\overline{p}_i,x_i^k))$, $\forall i\in [n], k\in \{0,1,\ldots,N-1\}$
            \STATE Obtain rating from users $x^N$
            \STATE CTR $y^N = \frac{1}{T+1}\sum_{k=N-T}^{N}c^k$
            \STATE Positions set $\mathcal{X}_{p}\leftarrow \textrm{append}[\overline{p}]$
            \STATE Steady-state opinions set $\mathcal{X}_x\leftarrow\textrm{append}[x^N]$
            \STATE CTR set $\mathcal{X}_{y} \leftarrow\textrm{append}[y^N]$
        \ENDFOR
    \end{algorithmic}
\end{algorithm}

 For each training sample  $k\in[m]$, we provide news expressing a fixed random position $\overline{p}$ to the users for a period of $N$ time steps. The term $T$ is a design parameter, with $(N-T)$ representing a time instant at which opinions have reached a steady-state. After $N$ steps, we obtain the users' rating $x^{N}$ and compute their CTR for the position $\overline{p}$ based on the last $T$ time steps. This process is carried out over $m$ Monte Carlo trials, thus collecting $m$ training data points. We emphasize that, similarly to what happens in real life with recommender systems asking the users to rate their experiences, for each trial, we only require access to one sample from the user opinion dynamics, $x^N$. The real time opinions in each trial $\{x^k\}_{k=0}^{N-1}$ are not needed. 

\subsection{Proof of Lemma \ref{est_opinion}}\label{est_opinionproof}
The proof follows from \cite[Corollary 5.2]{tabuada} with the continuous function $\beta$ playing the same role as $f$ in \cite{tabuada}.

In order to provide an explicit upper-bound for $e_x$ we point out that the neural network makes use of three layers: the injection layer, whose map is represented by $u: \mathbb{R}^{n+1}\to\mathbb{R}^{n+2}$, the intermediate layer whose map is represented by $z:\mathbb{R}^{n+2}\to\mathbb{R}^{n+2}$ and the output layer, whose map is represented by $v: \mathbb{R}^{n+2}\to\mathbb{R}$, with $v(z) = v^\top z + v_0$. Given that the input and output layers are linear maps, we can directly make use of \cite[Theorem 7]{marchi1} to state the following upper-bound $\forall i\in[n]$:
\begin{align*}
    |\beta_i(y_i,p) - \hat{\beta}_i(y_i,p)| \leq & 3\underset{y_i\in\mathcal{X}_{y,i},p\in\mathcal{X}_p}{\mathrm{sup}}|\beta_i(y_i,p) - \hat{\beta}_i(y_i,p)| + \\
    & 2\omega_{\beta_i}(\gamma_x) + |v_{0,i}|\gamma_x,
\end{align*}
 where $\underset{y\in\mathcal{X}_y,p\in\mathcal{X}_p}{\mathrm{sup}}|\beta_i(y,p) - \hat{\beta}_i(y,p)|$ refers to the maximum training error on each user $i$, $\omega_{\beta_i}(\gamma_x)$ refers to the minimum modulus of continuity of $\beta_i$ on the training set, $v_{0,i}$ represents the estimated bias weight $v_0$ from the intermediate layer to the output layer of the neural network for user $i$. Including the modelling error $\theta_x$ of $\beta$ from Assumption \ref{ANNassumption}, it is now possible to state the following:
\begin{align*}
    \lVert e_x\rVert \leq \sqrt{n}\Big[&3\underset{y\in\mathcal{X}_y,p\in\mathcal{X}_p}{\mathrm{sup}}\lVert\beta(y,p) - \hat{\beta}(y,p)\rVert_\infty +  \\
    &2\underset{i\in[n]}{\mathrm{sup}}\omega_{\beta_i}(\gamma_x) + \gamma_x\underset{i\in[n]}{\mathrm{sup}}|v_{0,i}|\Big]+\theta_x.
\end{align*}

\begin{arxiv}
\subsection{Proof of Lemma \ref{est_clickingbehaviour}}\label{est_clickingbehaviourproof}
Before making use of \cite[Corollary 5.2]{tabuada} to prove the upper-bound on $e_y$, it is to be noted that the arguments of $g$ and $\hat{g}$, in the definition of $e_y$, are different. Thus, we re-write $e_y$ as $ e_y= \hat{g}(p,\hat{x}) - g(p,\hat{x}) + g(p,\hat{x}) - g(p,h(p,d))$ and note the following:

\footnotesize
\begin{align*}
    \lVert e_y \rVert & \underset{(a)}{\leq} \lVert\hat{g}(p,\hat{x}) - {g}(p,\hat{x})\rVert + \lVert g(p,h(p,d) + e_x) - g(p,h(p,d)) \rVert \\
    & \underset{(b)}{\leq} \lVert\hat{g}(p,\hat{x}) - {g}(p,\hat{x})\rVert + \lVert\nabla_x g^\top(p,h(p,d)) e_x\rVert + \alpha_y\\
    & \underset{(c)}{\leq} \lVert\hat{g}(p,\hat{x}) - {g}(p,\hat{x})\rVert + M_x\epsilon_x + \alpha_y.
\end{align*}
\normalsize
  In $(a)$, we make use of the identity $\hat{x} = h(p,d) + e_x$, with the state estimation error $e_x$ defined in Lemma \ref{est_opinion}. In $(b)$, we use the Taylor series expansion on $g(p,h(p,d)+e_x)$ around $x$. In addition, the higher-order terms in the expansion of $g(p,h(p,d)+e_x)$ are upper-bounded by the modelling error on the clicking behaviour, i.e. $\rVert\mathcal{O}(g(p,x))\lVert \leq \alpha_{y}$. In $(c)$, we use Assumption \ref{measurementmodelassumpt}, i.e. $g(p,x)$ is $M_x$-Lipshitz with respect to $x$ and the fact that the opinion estimation error is upper-bounded with $\lVert e_x \rVert \leq \epsilon_x$ from Lemma \ref{est_opinion}.
  
  The existence of an upper-bound on $\lVert\hat{g}(p,\hat{x}) - {g}(p,\hat{x})\rVert$ is similar to the one in Lemma \ref{est_opinion}. Thus, we directly make use of \cite[Theorem 7]{marchi1} to state the following upper-bound $\forall i\in[n]$:

  \footnotesize
  \begin{align*}
    |g_i(p_i,\hat{x}_i) - \hat{g}_i(p_i,\hat{x}_i)| \leq & 3\underset{p_i\in\mathcal{X}_{p,i},x_i\in\mathcal{X}_{x,i}}{\mathrm{sup}}|g_i(p_i,x_i) - \hat{g}_i(p_i,x_i)| + \nonumber\\ & 2\omega_{g_i}(\gamma_y) + |w_{0,i}|\gamma_y,
\end{align*}
\normalsize
where $\underset{p_i\in\mathcal{X}_{p,i},x_i\in\mathcal{X}_{x,i}}{\mathrm{sup}}|g_i(p_i,x_i) - \hat{g}_i(p_i,x_i)|$ refers to the maximum training error on each user $i$, $\omega_{g_i}(\gamma_y)$ refers to the modulus of continuity of $g_i$ on $\mathcal{Y}$, $w_{0,i}$ represents the estimated bias weight $w_0$ from the penultimate layer to the output layer of the neural network for user $i$. It is now possible to state the following:
\begin{align*}
    \lVert e_y\rVert \leq \sqrt{n}\Big[&3\underset{p\in\mathcal{X}_p,x\in\mathcal{X}_x}{\mathrm{sup}}\lVert g(p,x) - \hat{g}(p,x)\rVert_\infty +  \nonumber\\
    &2\underset{i\in[n]}{\mathrm{sup}}\omega_{g_i}(\gamma_y) +\gamma_y\underset{i\in[n]}{\mathrm{sup}}|w_{0,i}|\Big]+M_x\epsilon_x + \alpha_y
\end{align*}


\end{arxiv}

\subsection{Hyper-parameter Tuning}\label{HPtuning}
In order to tune the process noise covariance $Q^k$ we rely on heuristics and qualitative observations.
 Empirical results on the FJ model suggests that the sensitivity matrix is non-negative. This leads to the condition  $\ell^k + w^k\geq 0_{n^2}, \forall k\in\mathbb{N}_0$ in the random walk \eqref{processmodelKF}. Therefore, to ensure a non-negative sensitivity, it is required to have the condition $\lVert w^k\rVert_\infty \leq \mathrm{min}\{\ell^k\}$. This condition is satisfied by choosing an $M$ large enough such that
\begin{equation}\label{sigmaq}
    \sigma_q^k = \underset{i:\hat{\ell}_i^k\neq0}{\mathrm{min}}\hat{\ell}_i^k/M.
\end{equation}
Here, $M$ is a tuning parameter, to be selected ad hoc. Note in fact, that setting $M$ arbitrarily large, the bound $\lVert w^k\rVert_\infty \leq \mathrm{min}\{\ell^k\}$ becomes looser, but this gives more trust in the process model, resulting in poor learning performances in the posterior estimates.

The measurement noise $v$, is tuned by maximizing the measurement likelihood as in \cite[Section 3D]{Abbeel2005DiscriminativeTO}. To this end, we solve the following optimization problem online:\
\begin{align}
   & (\sigma_r^k)^2 = \underset{(\sigma_r)^2}{\mathrm{arg}} {\mathrm{max}} \ P\Big(\Delta {x}_{\rm ss}^{k+1,\tau_i+1} \ | \ \Delta \Tilde{p}^{k,\tau_i} \Big) \label{MLE} 
\end{align}
The intuition behind \eqref{MLE} is to find the noise characteristics of $v$ that best represent the measurement model in \eqref{measurementmodelKF} using the input $\Delta p$ and output $\Delta {x}$ information. Note that, since the true measurements $\Delta x$ are unavailable, they are replaced with~$\Delta\hat{x}$.

\subsection{Proof of Lemma \ref{closeness}} \label{closenessproof}
    Due to Assumption \ref{measurementmodelassumpt} and Lemma \ref{Lsmoothlemma}, for any $p,x_1,x_2\in\mathbb{R}^n$ it holds that
\footnotesize
\begin{equation*}
    \varphi^{\mathrm{ctr}}(p,x_1) - \varphi^{\mathrm{ctr}}(p,x_2) - (\nabla_x \varphi^{\mathrm{ctr}})^\top (p,x_2) (x_1-x_2) \leq \frac{1}{2}L_x\lVert x_1-x_2\rVert^2.
\end{equation*}
\normalsize
In the above equation, we replace $x_1,x_2$ with $x+\mu e_i$ and $x$, respectively, so that
\begin{equation} \label{closeness_1}
    \varphi^{\mathrm{ctr}}(p,x + \mu e_i) - \varphi^{\mathrm{ctr}}(p,x) - \mu(\nabla_x \varphi^{\mathrm{ctr}})^\top (p,x)e_i \leq \frac{1}{2}L_x\mu^2.
\end{equation}
\normalsize
We now add and subtract $\hat{\varphi}^{\mathrm{ctr}}(p,\hat{x} + \mu e_i)$ and $\hat{\varphi}^{\mathrm{ctr}}(p,\hat{x})$ in the above equation and we note that $\hat{\varphi}^{\mathrm{ctr}}(p,\hat{x}) := -1_n^\top \hat{g}(p,\hat{x})$ and $\varphi^{\mathrm{ctr}}(p,x) := -1_n^\top g(p,x)$. Using these definitions in \eqref{closeness_1}, we obtain:
\begin{align}
    \hat{\varphi}^{\mathrm{ctr}}(p,\hat{x} + \mu e_i)& - \hat{\varphi}^{\mathrm{ctr}}(p,\hat{x}) \leq 1_n^\top[e_y(p,x) + e_y(p,x+\mu e_i)] \nonumber\\
    & +\mu(\nabla_x \varphi^{\mathrm{ctr}}(p,x))^\top e_i + \frac{1}{2}L_x\mu^2, \label{closeness_2}
\end{align}
 \normalsize where $e_y(p,x)$ is the clicking behaviour estimation error as defined in \eqref{upperbound_ANN_c}. Dividing both sides by $\mu$ in \eqref{closeness_2} and using the gradient estimate definition in \eqref{gradapprox_x}, we obtain 
\begin{align*}
\Big(\nabla_x\hat{\varphi}^{\mathrm{ctr}}\Big)_i(p,\hat{x})& - \Big(\nabla_x{\varphi}^{\mathrm{ctr}}\Big)_i(p,x) \leq  \frac{L_x\mu}{2} + \nonumber\\
    &\frac{1}{\mu}1_n^\top(e_y(p,x) + e_y(p,x+\mu e_i)).\label{closeness_3}
\end{align*}
\normalsize for all $ i\in[n]$. Taking the modulus on both sides, we obtain $|1_n^\top(e_y(p,x) + e_y(p,x+\mu e_i))| \leq \sqrt{n}\epsilon_g$ using the Cauchy-Schwartz inequality and the upper-bound on the clicking behaviour estimation error $\lVert e_y \rVert \leq \epsilon_g$ as in \eqref{upperbound_ANN_c}. We now write
\begin{equation*}
\Big|\Big(\nabla_x\hat{\varphi}^{\mathrm{ctr}}\Big)_i(p,\hat{x}) - \Big(\nabla_x {\varphi}^{\mathrm{ctr}}\Big)_i(p,x)\Big| \leq \frac{1}{2}L_x\mu + 2\frac{\sqrt{n}\epsilon_g}{\mu},
\end{equation*}
from which \eqref{gradapprox_x} follows. Analogous reasoning is followed for the gradient estimation error with respect to $p$. To obtain the tight upper-bound, we use first-order optimality conditions with respect to $\mu$ on the term $\frac{1}{2}L_x\mu + 2\frac{\sqrt{n}\epsilon_g}{\mu}$, thus obtaining $\mu^* = 2n^{1/4}\sqrt{\epsilon_g/L_x}$. Since the second-order derivative of the term $\frac{1}{2}L_x\mu + 2\frac{\sqrt{n}\epsilon_g}{\mu}$ is strictly positive, $\mu^*$ is the smoothing parameter that provides the lowest upper-bound on the gradient estimation error.

\subsection{Proof of Theorem \ref{sensitivityconvergence}} \label{sensitivityconvergenceproof}
     The Kalman filter is uniformly asymptotically stable provided the pairs $(I_{n^2},\sqrt{Q^k}), (I_{n^2},\Delta\Tilde{p}^{k,\tau_i})$ are uniformly completely controllable and observable, respectively \cite[{Theorem 7.4}]{andrew}. By design (see Appendix \ref{HPtuning}), the process noise covariance $Q^k$ is positive definite $\forall k\in\mathbb{N}_0$, from which the matrix pair $(I_{n^2},\sqrt{Q^k})$ is uniformly completely controllable.

    To prove observability of  $(I_{n^2},\Delta\Tilde{p}^{k,\tau_i})$, we consider the following observability Gramian matrix \cite[Chapter 7.5]{andrew}:
    \begin{equation}\label{observability}
        W_O(l_1,l_2) := \sum_{l=l_1}^{l_2} (\Delta\Tilde{p}^{\tau_{i-l+1},\tau_{i-l}})^\top (R^{\tau_{i-l+1}})^{-1}(\Delta\Tilde{p}^{\tau_{i-l+1},\tau_{i-l}}).
    \end{equation}
    The matrix pair $(I_{n^2},\Delta\Tilde{p}^{k,\tau_i})$ is said to be uniformly completely observable if $\exists \  \alpha,\beta>0$ and an $N\in \mathbb{N}_0: 0 < \alpha I_{n^2} \leq W_O(k-N,k) \leq \beta I_{n^2}$.
    Since $\Delta \Tilde{p}^{t_1,t_2}:= \Delta p^{t_1,t_2} \otimes I_n$, we write $W_O(l_1,l_2) = \sum_{l=l_1}^{l_2} (\Delta{p}^{\tau_{i-l+1},\tau_{i-l}})(\Delta{p}^{\tau_{i-l+1},\tau_{i-l}})^\top \otimes (R^{\tau_{i-l+1}})^{-1}$. With an appropriate choice of $M$ in \eqref{sigmaq}, we have $\alpha_1 I_n \leq R^k \leq \alpha_2I_n, \forall k\in\mathbb{N}_0$, with $\alpha_1>0$ and $\alpha_2<\infty$, leading to $\frac{1}{\alpha_2}\leq(R^{\tau_{i-l+1}})^{-1}\leq\frac{1}{\alpha_1}$. Moreover, the persistency of excitation assumption (see Assumption \ref{POEassumption}) on the inputs $\Delta p$ enable us to state that $\exists \  S\in\mathbb{N}_0: \beta_1 I_n \leq \sum_{l=k-S}^k(\Delta{p}^{\tau_{i-l+1},\tau_{i-l}})(\Delta{p}^{\tau_{i-l+1},\tau_{i-l}})^\top \leq \beta_2I_n$, with $\beta_1>0$ and $\beta_2 = (S+1)\underset{l\in[k-S,k]}{\mathrm{sup}} \ \lVert\Delta{p}^{\tau_{i-l+1},\tau_{i-l}}\rVert^2\leq 4n(S+1)$. This condition along with the inequality on $R^k$ leads to the result $0<\frac{\beta_1}{\alpha_2}I_{n^2}\leq W_O(k-S,k)\leq \frac{\beta_2}{\alpha_1}I_{n^2}$. Thus, the matrix pair $(I_{n^2},\Delta\Tilde{p}^{k,\tau_i})$ is uniformly completely observable and it is now possible to state that the Kalman filter is uniformly asymptotically stable.

    We can now derive an explicit analytical expression for the upper-bound on the bias $\lVert\mathbb{E}[e^k]\rVert$ and the variance $\mathbb{E}[\lVert e^k \rVert^2]$.  By making use of \eqref{processmodelKF} and \eqref{KFposterior_state} on $\ell^k$ and $\hat{\ell}^k$, by adding and subtracting $\Delta x_{\rm ss}^{k+1,\tau_i+1}$and by making use of \eqref{processmodelKF} and \eqref{measurementmodelKF} on $\Delta x_{\rm ss}^{k+1,\tau_i+1}$, one gets
     \begin{align}
     e^k &= \big(I_{n^2} - \zeta^kK^{k-1}\Delta\Tilde{p}^{k,\tau_i}\big)(e^{k-1}+w^{k-1}) - \zeta^kK^{k-1}v^k  \nonumber\\
        & + \ \zeta^kK^{k-1}\Delta e_x^{\tau_i+1,k+1}. \label{errorsensitivity_1}
     \end{align}
    
    Considering a trigger event at $k$, i.e. $\zeta^k=1$, we trace \eqref{errorsensitivity_1} back to the previous trigger instant at $\tau_i$, thus obtaining:
    \begin{align}
        e^{k} =& \big(I_{n^2} - K^{k-1}\Delta\Tilde{p}^{k,\tau_i}\big)e^{\tau_i} + \big(I_{n^2} - K^{k-1}\Delta\Tilde{p}^{k,\tau_i}\big)w^{k-1} + \nonumber\\
        & \sum_{t=\tau_i}^{k-2}w^t - K^{k-1}v^k + K^{k-1}\Delta e_x^{\tau_i+1,k+1}. \label{error_sensitivity_2}
    \end{align}
    Taking the expectation on both sides of the above equation and using Assumption \ref{IIDassumption}, we have:
    \begin{equation}\label{bias_1}
        \mathbb{E}[e^k] = \big(I_{n^2} - K^{k-1}\Delta\Tilde{p}^{k,\tau_i}\big)\mathbb{E}[e^{\tau_i}] + K^{k-1}\mathbb{E}[\Delta e_x^{\tau_i+1,k+1}]. 
    \end{equation}
    Since we have already established uniform complete observability, the Kalman filter gain $K^k$ settles to a steady-state value $K_s$ as $k\to\infty$. Thus, we can state that $\exists \ K_m: \lVert K^k\rVert \leq K_m, \forall \  k\in\mathbb{N}^0$. Further, denoting $\psi(\tau_{i+1},\tau_i):= I_{n^2} - K^{\tau_{i+1}-1}\Delta\Tilde{p}^{\tau_{i+1},\tau_i}$ as the state transition matrix for the Kalman filter error $e^k$, uniform asymptotic stability leads to $\lVert \psi(\tau_{i+1},\tau_i) \rVert \leq c_1 \xi^{c_2T}$ for some $c_1,c_2>0, \xi\in (0,1) : c_1 \xi^{c_2T}<1$, where $\tau_{i+1}=k$ is the new trigger time instant and $T$ is the time period. Thus, taking the norm on both sides in \eqref{bias_1}, we have
       $ \lVert \mathbb{E}[e^{\tau_{i+1}}] \rVert \leq c_1\xi^{c_2T}\lVert\mathbb{E}[e^{\tau_i}]\rVert + 2K_m\epsilon_x$, 
    where we note that $\mathrm{sup} \ \lVert\mathbb{E}[\Delta e_x^{\tau_i+1,k+1}]\rVert \leq \mathbb{E}[\lVert \Delta e_x^{\tau_i+1,k+1} \rVert] \leq 2\epsilon_x$ from Lemma \ref{est_opinion}. Hence, we derive the bias equation
    \begin{equation*}
        \lVert \mathbb{E}[e^{\tau_{i+1}}] \rVert \leq \big(c_1\xi^{c_2T}\big)^{|\mathcal{T}|}\lVert\mathbb{E}[e^{0}]\rVert + 2K_m\epsilon_x\frac{1-(c_1\xi^{c_2T})^{|\mathcal{T}|}}{1-c_1\xi^{c_2T}},
    \end{equation*}
    where, $\mathcal{T}:= \{\tau_0,\tau_1,\dots,\tau_i\}$ is the set of all trigger time instances upto $\tau_i$. Thus, we have the asymptotic result on the bias $\lim_{|\mathcal{T}|\to\infty}\mathbb{E}[e^{\tau_{i}}] \leq 2K_m\epsilon_x/(1-c_1\xi^{c_2T})$.

    We now derive the upper-bound on the variance $\mathbb{E}[\lVert e^k \rVert^2]$. Taking the expectation of the  norm squared on both sides in \eqref{errorsensitivity_1}, we obtain:
    \begin{align*}
        \mathbb{E}[\lVert e^k \rVert^2] \underset{(a)}{=} \mathbb{E}[\lVert & (I_{n^2} - \zeta^kK^{k-1}\Delta\Tilde{p}^{k,\tau_i})(e^{k-1} + w^{k-1}) - \nonumber\\
        &\zeta^k K^{k-1}\big(v^k + \Delta e_x^{k+1,\tau_i+1}\big)\rVert^2] \nonumber\\
        \underset{(b)}{\leq} \lVert I_{n^2}& - \zeta^kK^{k-1}\Delta\Tilde{p}^{k,\tau_i} \rVert^2(\mathbb{E}[\lVert e^{k-1} \rVert^2] + (\sigma_q^{k-1})^2) \nonumber\\
        +\lVert&\zeta^k K^{k-1}\rVert^2\big((\sigma_r^k)^2 + 2\epsilon_x^2\big)
    \end{align*}
     In $(b)$, we expand the norm and use Assumption \ref{IIDassumption} to state that the expectation on the cross-coupled terms are all zero. We then use $\mathbb{E}[\lVert w^{k-1}\rVert^2] = (\sigma_q^{k-1})^2$ and $\mathbb{E}[\lVert v^k\rVert^2] = (\sigma_r^k)^2$. Further, Lemma \ref{est_opinion} allows us to state that $\mathbb{E}[\lVert \Delta e_x^{k+1,\tau_i+1} \rVert^2] \leq 2\epsilon_x^2$.
    Considering a trigger at time $k$, we have:
    \begin{align}\label{sensitivityerror_2}
        \mathbb{E}[\lVert e^{\tau_{i+1}} \rVert^2] \leq & \lVert I_{n^2} - K^{\tau_{i+1}-1}\Delta\Tilde{p}^{\tau_{i+1},\tau_i} \rVert^2\mathbb{E}[\lVert e^{\tau_i} \rVert^2] + \nonumber\\
        &(T-1)\overline{\sigma}_q^2 + \lVert K^{\tau_{i+1}-1}\rVert^2\big(\overline{\sigma}_r^2 + 2\epsilon_x^2\big),
    \end{align}
    where $\overline{\sigma}_q^2 = \underset{t\in\mathbb{N}_0}{\mathrm{sup}} ({\sigma}_q^t)^2$ and $\overline{\sigma}_r^2 = \underset{t\in\mathbb{N}_0}{\mathrm{sup}} ({\sigma}_r^t)^2$ with $\sigma_q,\sigma_r$ estimated using \eqref{sigmaq} and \eqref{MLE}, respectively. Tracing back \eqref{sensitivityerror_2} recursively to $k=0$, we obtain the following relation:
    \begin{align*}\label{sensitivityerror_3}
        \mathbb{E}[\lVert e^{\tau_{i+1}} \rVert^2] \leq & (c_1\xi^{c_2T})^{2|\mathcal{T}|}\mathbb{E}[\lVert e^0 \rVert^2] + \nonumber\\
        &\frac{1-(c_1\xi^{c_2T})^{2|\mathcal{T}|}}{1-c_1\xi^{c_2T}}\big[ (T-1)\overline{\sigma}_q^2 + K_m^2\big(\overline{\sigma}_r^2 + 2\epsilon_x^2\big)\big],
    \end{align*}
     We now have the following asymptotic result on the variance:
    \begin{equation*}
        \lim_{|\mathcal{T}|\to\infty}\mathbb{E}[\lVert e^{\tau_i} \rVert^2] \leq \frac{1}{1-c_1\xi^{c_2T}}\big[ (T-1)\overline{\sigma}_q^2 + K_m^2\big(\overline{\sigma}_r^2 + 2\epsilon_x^2\big)\big],
    \end{equation*}



\subsection{Proof of Theorem \ref{OFOconvergence}}
Before describing the proof, we state the following supporting lemma and remark.
\begin{lemma}[Projections with smooth functions \cite{reddi2016fast}]\label{supplementary_1}
    Let $y = \Pi_{[-1,1]^n}[x - \eta u]$ with $y,x,u\in\mathbb{R}^n$. Then, the following inequality holds:
    \begin{align*}
        \varphi(y) \leq & \varphi(z) + \langle y-z,\Phi(x)-u\rangle + \Big[\frac{L'}{2} - \frac{1}{2\eta}\Big]\lVert y-x\rVert^2 + \\
        & \Big[\frac{L'}{2} + \frac{1}{2\eta}\Big]\lVert z-x\rVert^2 - \frac{1}{2\eta}\lVert y-z\rVert^2, \quad \forall  z\in\mathbb{R}^n
    \end{align*}
    where $\varphi$ is the cost function to be minimized and $\Phi$ its gradient. Further, $L'$ and $\eta$ are the smoothness factor of $\varphi$ and the step-size of the gradient descent algorithm, respectively.
\end{lemma}
\begin{proof}\label{supplementary_1_proof}
    The proof is given in \cite[Lemma 2]{reddi2016fast}.
\end{proof}\label{OFOconvergenceproof}
\begin{remark}[Composite gradient lipschitzness]\label{Lsmooth_full}
    We observe that the Lipschitz and smoothness constant for $\varphi^{\mathrm{pol}}(x)$ are $2\sqrt{n}$ and $2$, respectively. Using Assumptions \ref{ass:ODasm}(iii), \ref{measurementmodelassumpt}, the composite gradient $\Phi(p)$ is thus $L'-$Lipschitz with respect to $p$, where $L'=L_p + L^2(L_x + 2)$. Therefore, the composite cost function $\varphi(p,h(p,d))$ is $L'-$smooth with respect to $p$. 
\end{remark}

We now state the following gradient update in the case where the gradients, sensitivity and opinions are known:
\begin{equation*}
    \overline{p}^{k+1} = \Pi_{[-1,1]^n}[p^k - \eta\Phi(p^k,h(p^k,d))],
\end{equation*}
where $\Phi(p,h(p,d))$ is the composite gradient.
The above will serve as a benchmark to investigate the stationarity of our algorithm. We are now in the position to prove the inequalities in \eqref{OFOconvergenceequation}. 
To do so, we use Lemma \ref{supplementary_1} with $y = \overline{p}^{k+1}$, $x=p^k$ and $u =\Phi(p^k,h(p^k,d))$ is the composite gradient. Choosing $z = p^k$ and taking the expectation on both sides, we obtain:
\begin{equation}\label{OFO_convergence_step1}
        \mathbb{E}\Big[\varphi(\overline{p}^{k+1})\Big] \leq \mathbb{E}\Big[\varphi(p^k) + \Big(\frac{L'}{2} - \frac{1}{\eta}\Big)\lVert\overline{p}^{k+1}-p^k\rVert^2\Big].
\end{equation}
We now define the update step with our algorithm:
\begin{equation*}
        {p}^{k+1} = \Pi_{[-1,1]^n}[p^k - \eta\zeta^k\hat{\Phi}^k],
\end{equation*}
where, the composite gradient estimate $\hat{\Phi}^k$ is given in \eqref{compositegradientestimate}. We use Lemma \ref{supplementary_1} with $y = p^{k+1}$, $x=p^k$ and $u = \zeta^k\hat{\Phi}^k$. Choosing $z = \overline{p}^{k+1}$ and taking the expectation on both sides of the inequality, we obtain:
\footnotesize
\begin{align}
        \mathbb{E}\Big[\varphi(p^{k+1})\Big] \leq  \mathbb{E}\Big[ & \varphi(\overline{p}^{k+1}) + \Big(\frac{L'}{2} - \frac{1}{2\eta}\Big)\lVert p^{k+1}-p^k\rVert^2 + \nonumber\\
        & \langle p^{k+1}-\overline{p}^{k+1},\Phi(p^k,h(p^k,d)) - \zeta^k\hat{\Phi}^k\rangle + \label{OFO_convergence_step2}\\
        & \Big(\frac{L'}{2} + \frac{1}{2\eta}\Big)\lVert\overline{p}^{k+1}-p^k\rVert^2 -\frac{1}{2\eta}\lVert p^{k+1}-\overline{p}^{k+1}\rVert^2\Big] \notag
\end{align}
\normalsize We now add inequalities \eqref{OFO_convergence_step1} and \eqref{OFO_convergence_step2} to obtain:
\footnotesize
\begin{align}
        \mathbb{E}\Big[\varphi(p^{k+1})\Big] \leq \mathbb{E}\Big[ & \varphi(p^k) + \Big(\frac{L'}{2} - \frac{1}{2\eta}\Big)\lVert p^{k+1}-p^k\rVert^2 + \nonumber\\
        & \underset{T_1}{\underbrace{\langle p^{k+1}-\overline{p}^{k+1},\Phi(p^k,h(p^k,d)) - \zeta^k\hat{\Phi}^k\rangle}} +  \label{OFO_convergence_step3}\\
        & \Big({L'} - \frac{1}{2\eta}\Big)\lVert\overline{p}^{k+1}-p^k\rVert^2 -\frac{1}{2\eta}\lVert p^{k+1}-\overline{p}^{k+1}\rVert^2\Big].\notag
\end{align}

\normalsize
We now focus on the term $T_1$. Using Cauchy-Schwartz relation and the fact that the geometric mean of two non-negative real numbers is always less than its arithmetic mean, we obtain the following:
$ T_1 \leq \lVert p^{k+1}-\overline{p}^{k+1}\rVert \ \lVert\Phi(p^k,h(p^k,d)) - \zeta^k\hat{\Phi}^k\rVert
\leq \frac{1}{2\eta}\lVert p^{k+1}-\overline{p}^{k+1}\rVert^2 + \frac{\eta}{2}\lVert\Phi(p^k,h(p^k,d)) - \zeta^k\hat{\Phi}^k\rVert^2.$
 Using the above inequality in \eqref{OFO_convergence_step3}, we obtain:

 \footnotesize
    \begin{align}
    &\mathbb{E}\Big[\varphi(p^{k+1})\Big] \leq \mathbb{E}\Big[  \varphi(p^k) + \underset{T_2}{\underbrace{\Big(\frac{L'}{2} - \frac{1}{2\eta}\Big)\lVert p^{k+1}-p^k\rVert^2}} + \label{OFO_convergence_step4}\\
    & \frac{\eta}{2}{{\lVert\Phi(p^k,h(p^k,d)) - \zeta^k\hat{\Phi}^k\rVert^2}} + \Big({L'} - \frac{1}{2\eta}\Big)\eta^2\lVert\mathcal{G}(p^k)\rVert^2 \Big].\notag
    \end{align}
    \normalsize In the above inequality, we used the definition of fixed-point residual mapping according to \eqref{gradientmapping}.
    
    We now assume that there is a trigger at time instant $k$, i.e. $\zeta^k =1$. To obtain a feasible upper-bound on $\mathbb{E}[\lVert\mathcal{G}(p^k)\rVert^2]$, we need $L' - 1/2\eta<0$. Thus, the step-size is constrained with $\eta \in (0,\frac{1}{2L'})$. With this constraint, we have $T_2\leq 0$. This leads to the following:
    \begin{align}
        \mathbb{E}\Big[\lVert\mathcal{G}(p^k)\rVert^2\Big] &\leq \frac{2}{\eta(1-2\eta L')}\Big\{ \mathbb{E}\Big[\varphi(p^k)\Big] - \mathbb{E}\Big[\varphi(p^{k+1})\Big] + \nonumber\\
        & \frac{\eta}{2}\underset{T_3}{\underbrace{\mathbb{E}\Big[\lVert\Phi(p^k,h(p^k,d)) - \hat{\Phi}^k\rVert^2\Big]}}\Big\}. \label{OFO_convergence_step5}
    \end{align}
\normalsize We now analyze the term $T_3$. For the sake of convenience, we drop the arguments of the gradient $p,x$ and time argument $k$ in the gradient terms. Thus, we denote $\nabla_p\varphi^{\rm ctr}=\nabla_p\varphi^{\rm ctr}(p^k,h(p^k,d))$, $\nabla_x\varphi^{\rm ctr}=\nabla_x\varphi^{\rm ctr}(p^k,h(p^k,d))$,
    $\nabla_p\hat{\varphi}^{\rm ctr}=\nabla_p\hat{\varphi}^{\rm ctr}(p^k,\hat{x}^{k+1})$ and $\nabla_x\hat{\varphi}^{\rm ctr}=\nabla_x\hat{\varphi}^{\rm ctr}(p^k,\hat{x}^{k+1})$. Using the definitions of $\Phi(p^k,h(p^k,d)), \hat{\Phi}^k$, we have:

    \footnotesize
    \begin{align}
        T_3 = \mathbb{E}\Big[ & \lVert\nabla_p\varphi^{\rm ctr}-\nabla_p\hat{\varphi}^{\rm ctr} + H^{k\top}\nabla_x\varphi^{\rm ctr} - \hat{H}^{k\top}\nabla_x\hat{\varphi}^{\rm ctr} + \label{T3eqn1} \\
        & \gamma H^{k\top}\nabla_x\varphi^{\rm pol}(h(p^k,d))
         -\gamma \hat{H}^{k\top}\nabla_x\varphi^{\rm pol}(\hat{x}^{k+1}) + \frac{w_{\rm pe}^k}{\eta}\rVert^2\Big],\notag 
    \end{align}
    \normalsize where $H^k=\nabla_p h(p^k,d)$ is the true sensitivity and $\hat{H}^k$ is its estimate at time $k$. We now analyze the upper-bound on $T_3$:

    \footnotesize
    \begin{align}
        T_3 \overset{(a)}{\leq} 6 \Big\{&\lVert\nabla_p\varphi^{\rm ctr}-\nabla_p\hat{\varphi}^{\rm ctr}\rVert^2 + \lVert(H^k)^\top(\nabla_x\varphi^{\rm ctr}-\nabla_x\hat{\varphi}^{\rm ctr})\rVert^2 + \nonumber\\
        & \frac{\sigma_{\rm pe}^2}{\eta^2} + \gamma^2\lVert(H^k)^\top(\nabla_x\varphi^{\rm pol}(h(p^k,d))-\nabla_x{\varphi}^{\rm pol}(\hat{x}^{k+1}))\rVert^2 + \nonumber\\
        & \big(\gamma^2\lVert\nabla_x\varphi^{\rm pol}(\hat{x}^{k+1})\rVert^2 + \lVert\nabla_x\hat{\varphi}^{\rm ctr}\rVert^2\big)\mathbb{E}\Big[\lVert H^k-\hat{H}^k\rVert^2\Big]\Big\} \nonumber\\
        \overset{(b)}{\leq} 6 \Big\{& 4n^{3/2}\epsilon_g (L_p + L^2L_x) + \frac{\sigma_{\rm pe}^2}{\eta^2} + \nonumber\\
        & \gamma^2\lVert(H^k)^\top(\nabla_x\varphi^{\rm pol}(h(p^k,d))-\nabla_x{\varphi}^{\rm pol}(\hat{x}^{k+1}))\rVert^2 + \nonumber\\
        & \big(\gamma^2\lVert\nabla_x\varphi^{\rm pol}(\hat{x}^{k+1})\rVert^2 + \lVert\nabla_x\hat{\varphi}^{\rm ctr}\rVert^2\big)\mathbb{E}\Big[\lVert H^k-\hat{H}^k\rVert^2\Big]\Big\} \nonumber\\
        \overset{(c)}{\leq} 6 \Big\{& 4n^{3/2}\epsilon_g (L_p + L^2L_x) + \frac{\sigma_{\rm pe}^2}{\eta^2} + 4\gamma^2L^2\lVert h(p^k,d) -\hat{x}^{k+1}\rVert^2 + \nonumber\\
        & \big(\gamma^2\lVert\nabla_x\varphi^{\rm pol}(\hat{x}^{k+1})\rVert^2 + \lVert\nabla_x\hat{\varphi}^{\rm ctr}\rVert^2\big)\mathbb{E}\Big[\lVert H^k-\hat{H}^k\rVert^2\Big]\Big\} \nonumber\\
        \overset{(d)}{\leq} 6 \Big\{& 4n^{3/2}\epsilon_g (L_p + L^2L_x) + \frac{\sigma_{\rm pe}^2}{\eta^2} + 4\gamma^2L^2\epsilon_x^2 + \nonumber\\
        & \big(\gamma^2\lVert\nabla_x\varphi^{\rm pol}(\hat{x}^{k+1})\rVert^2 + \lVert\nabla_x\hat{\varphi}^{\rm ctr}\rVert^2\big)\mathbb{E}\Big[\lVert H^k-\hat{H}^k\rVert^2\Big]\Big\} \nonumber\\
        \overset{(e)}{\leq} 6 \Big\{& 4n^{3/2}\epsilon_g (L_p + L^2L_x) + \frac{\sigma_{\rm pe}^2}{\eta^2} + 4\gamma^2L^2\epsilon_x^2 + \nonumber\\
        & \big(4n\gamma^2 + \lVert\nabla_x\hat{\varphi}^{\rm ctr} \pm \nabla_x{\varphi}^{\rm ctr} \rVert^2\big)\mathbb{E}\Big[\lVert e^k\rVert^2\Big]\Big\} \nonumber\\
        \overset{(f)}{\leq} 6 \Big\{& 4n^{3/2}\epsilon_g (L_p + L^2L_x) + \frac{\sigma_{\rm pe}^2}{\eta^2} + 4\gamma^2L^2\epsilon_x^2 + \nonumber\\
        & \big(4n\gamma^2 + 2(M_x^2 + 4n^{3/2}\epsilon_gL_x)\big)\mathbb{E}\Big[\lVert e^k\rVert^2\Big]\Big\} \label{T3eqn2}
    \end{align}
    \normalsize To obtain $(a)$, we added and subtracted $(H^k)^\top\nabla_x\hat{\varphi}^{\rm ctr}$ and $\gamma(H^k)^\top \nabla_x \varphi^{\rm pol}(\hat{x}^{k+1})$ inside the norm in \eqref{T3eqn1}. We then used the fact that $\mathbb{E}\Big[\lVert\sum_{j=1}^m r_j\rVert^2\Big]\leq m\sum_{j=1}^m\mathbb{E}\Big[\lVert r_j\rVert^2\Big]$. We also make use of the fact that $\mathbb{E}[\lVert w_{\rm pe}^k\rVert^2]=\sigma_{\rm pe}^2$. In $(b)$, we made use of Lemma \ref{closeness} for the gradient estimate accuracy on $\varphi^{\rm ctr}$ and Assumption \ref{ass:ODasm}(iii) for Lipschitz condition on $h(p,d)$ to arrive at the term $4n^{3/2}\epsilon_g (L_p + L^2L_x)$. In $(c)$, we made use of Remark \ref{Lsmooth_full} for the smoothness condition on $\varphi^{\rm pol}$ and Assumption \ref{ass:ODasm}(iii) for Lipschitz condition on $h(p,d)$ to arrive at the term $4L^2\lVert x - \hat{x}\rVert^2$. In $(d)$, we made use of the fact that the norm of the steady-state opinion estimation error is upper-bounded by $\epsilon_x$. In $(e)$, we made use of Remark \ref{Lsmooth_full} for the Lipschitz condition on $\varphi^{\rm pol}$, thus arriving at the term $4n$. We also add and subtract the term $\nabla_x{\varphi}^{\rm ctr}$. Further, we use inequality $\lVert H^k-\hat{H}^k\rVert \leq \lVert H^k-\hat{H}^k\rVert_F = \lVert e^k \rVert$, where $\lVert\cdot\rVert_F$ refers to the Frobenius norm and $e^k$ is the sensitivity estimation error. In $(f)$, we made use of the fact that $\lVert a + b \rVert^2 \leq 2(\lVert a \rVert^2 + \lVert b \rVert^2)$. We then use Assumption \ref{measurementmodelassumpt} to state that $\lVert \nabla_x \varphi^{\rm ctr} \rVert \leq M_x$ and Lemma \ref{closeness} to state the upper-bound on $\lVert \nabla_x \varphi^{\rm ctr} - \nabla_x \hat{\varphi}^{\rm ctr} \rVert$. 

    We now use the inequality \eqref{T3eqn2} in \eqref{OFO_convergence_step5} and add the inequalities over the trigger time instances up to $k$, leading to telescopic cancellation. Thus, we have:

    \footnotesize
    \begin{align*}
        \sum_{\underset{l\leq k}{l\in\mathcal{T}}}\mathbb{E}[\lVert\mathcal{G}(p^l)\rVert^2&] \leq  \frac{2}{\eta(1-2\eta L')}\Big\{\mathbb{E}\Big[\varphi(p^0)\Big] - \mathbb{E}\Big[\varphi(p^{k+1})\Big]\Big\} + \\
        & \frac{6|\mathcal{T}|}{1-2\eta L'}\Big\{4n^{3/2}\epsilon_g(L_p+L^2L_x) + \frac{\sigma_{\rm pe}^2}{\eta^2}+ 4\gamma^2L^2\epsilon_x^2\Big\} + \nonumber\\
        & \frac{12\big[2n\gamma^2 + (M_x^2 + 4n^{3/2}\epsilon_g L_x)\big]}{1-2\eta L'} \sum_{\underset{l\leq k}{l\in\mathcal{T}}}\mathbb{E}\Big[\lVert e^l\rVert^2\Big]. \notag 
    \end{align*}
    \normalsize Since the positions do not change between two consecutive trigger time instances, it is sufficient to investigate convergence guarantees at the trigger time instances. 
    
    Using Theorem \ref{sensitivityconvergence} for the upper-bound on $\mathbb{E}[\lVert e^k \rVert^2]$, the summation of this term over the trigger instances leads to the following:
    \begin{align}\label{sensitivity_informal}
        \sum_{\underset{l\leq k}{l\in\mathcal{T}}}\mathbb{E}\Big[\lVert e^l\rVert^2\Big] \leq & |\mathcal{T}|C_f + \mathbb{E}\Big[\lVert e^0\rVert^2\Big] \sum_{\underset{l\leq k}{l\in\mathcal{T}}}(c_1\xi^{c_2T})^{2|\mathcal{T}|} \\ & |\mathcal{T}|C_f +\mathbb{E}\Big[\lVert e^0\rVert^2\Big]\Big(\frac{1-(c_1\xi^{c_2T})^{2|\mathcal{T}|}}{1-(c_1\xi^{c_2T})^2}\Big).\notag
    \end{align}
    \normalsize We start the algorithm with $p^0=0_n$, thus $\mathbb{E}\Big[\varphi(p^0)\Big] = \varphi(0)$. Further, $\exists \ \varphi^{*}\leq\mathbb{E}\Big[\varphi(p^k)\Big], \forall k\in\mathbb{N}_0$, i.e. $\varphi^{*}$ is a local optimal value. Thus, with the above formulations and \eqref{sensitivity_informal}, we obtain \eqref{OFOconvergenceequation}.

    \section*{References}
\bibliographystyle{myieeetr}
\bibliography{references}

\begin{tac}
\begin{IEEEbiography}[{\includegraphics[width=1in,height=1.25in,clip,keepaspectratio]{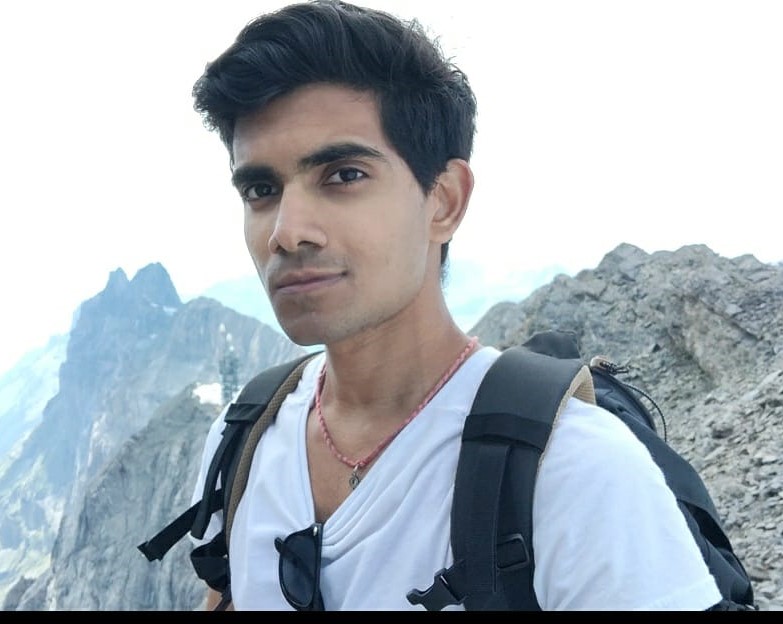}}]{Sanjay Chandrasekaran} recieved his B. Tech. degree in Instrumentation and Control and his MSc. degree in Information Technology and Electrical Engineering from the National Institute of Technology, Tiruchirappalli, India and ETH Zurich, Zurich, Switzerland in 2021 and 2023, respectively. His research interests lie in optimization-based control, learning-based control and network control systems.
\end{IEEEbiography}

\begin{IEEEbiography}[{\includegraphics[width=1in,height=1.25in,clip,keepaspectratio]{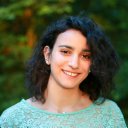}}]{Giulia De Pasquale} (Member, IEEE)
received the B.Sc. degree in Information Engineering, the M.Sc. degree in Control engineering and Ph.D. degree in Automatic Control in 2017, 2019, 2023, respectively, from the University of Padova, Padua, Italy. She is currently a post-doctoral researcher at the Automatic Control Laboratory, ETH Zürich, Zurich,
Switzerland. Her current research interests include modeling, analysis, and control of networked socio-technical systems.
\end{IEEEbiography}

\begin{IEEEbiography}[{\includegraphics[width=1in,height=1.25in,clip,keepaspectratio]{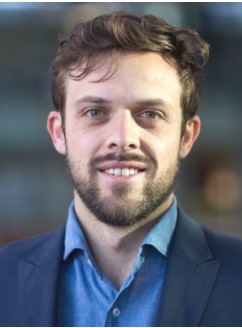}}]{Giuseppe Belgioioso} (Member, IEEE) is an Assistant Professor at the Division of Decision and Control Systems at KTH Royal Institute of Technology, Sweden. He received the Bachelor's degree in Information Engineering in 2012 and the Master's degree (cum laude) in Control Systems in 2015, both at the University of Padova, Italy. In 2020, he obtained the Ph.D. degree in Automatic Control at Eindhoven University of Technology (TU/e), The Netherlands. From 2021 to 2024, he was first a Postdoctoral researcher and then Senior Scientist at the Automatic Control Laboratory at ETH Z\"{u}rich, Switzerland.
His research lies at the intersection of optimization, game theory, and automatic control with applications in complex systems, such as electrical power grids and traffic networks.
\end{IEEEbiography}

\begin{IEEEbiography}[{\includegraphics[width=1in,height=1.25in,clip,keepaspectratio]{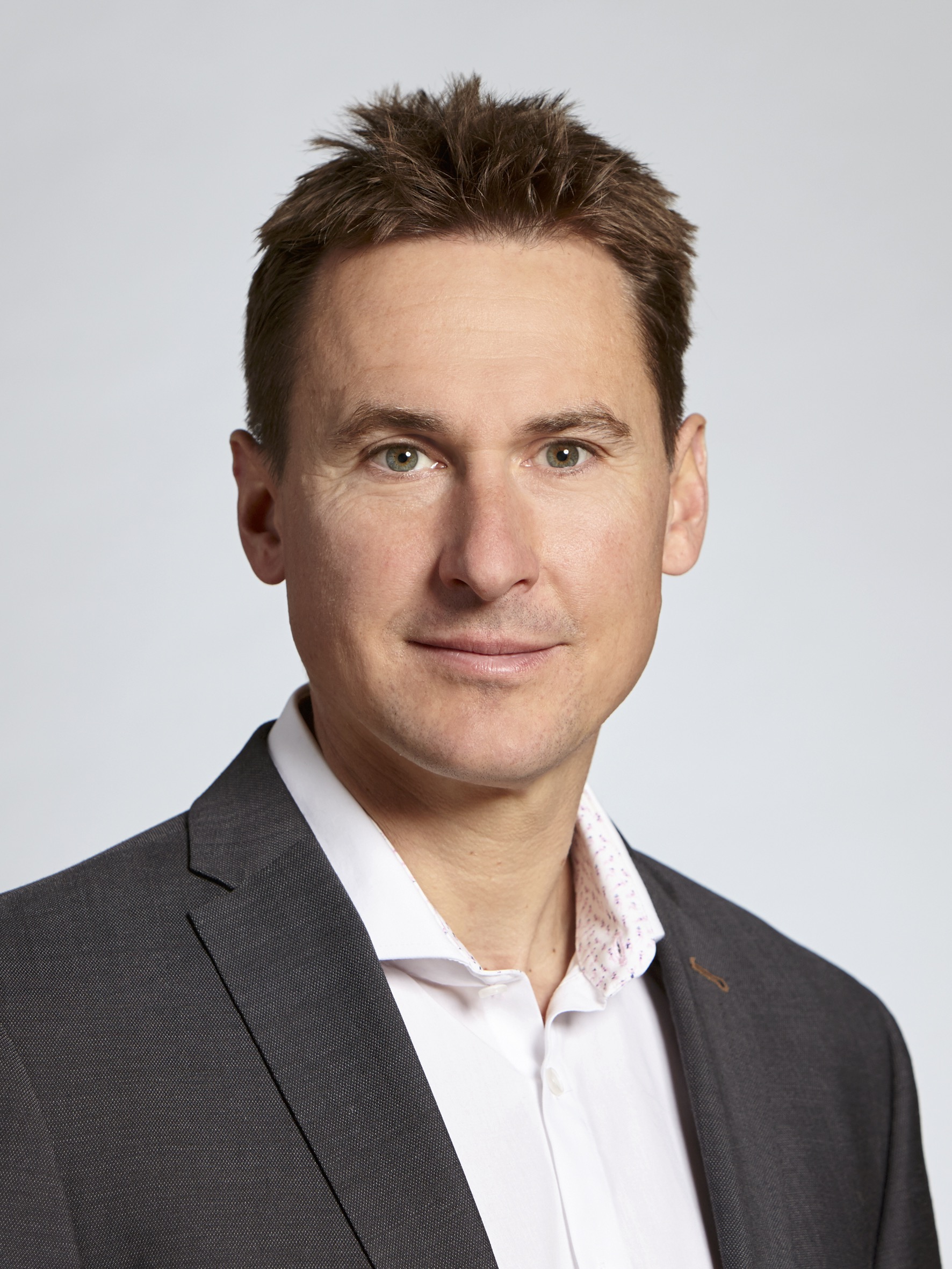}}]{Florian D\"orfler} (Senior Member, IEEE) received the Diploma in engineering cybernetics from the University of Stuttgart, Stuttgart, Germany, in 2008 and the Ph.D. degree in mechanical engineering from
the University of California at Santa Barbara, CA, USA, in 2013.
He is currently a Full Professor with the Automatic Control Laboratory, ETH Zürich, Zürich, Switzerland. His research
interests include control, optimization, and system theory with applications in
network systems, in particular electric power grids.
\end{IEEEbiography}
\end{tac}

\end{document}